\documentclass[a4paper]{scrartcl}
\usepackage[activate={true,nocompatibility},stretch=15]{microtype}
\usepackage[utf8]{inputenc}
\usepackage[english]{babel}
\usepackage{csquotes} 
\usepackage[bitstream-charter]{mathdesign}
\usepackage[T1]{fontenc}
\usepackage[EULERGREEK]{sansmath}
\usepackage{helvet}
\usepackage{xcolor}
\usepackage{pgf}
\DeclareUnicodeCharacter{2212}{-}

\usepackage{amsmath}
\usepackage{amsthm}
\usepackage[hidelinks,pdfusetitle,linktoc=all]{hyperref}
\usepackage[nameinlink,capitalize]{cleveref}
\newtheorem{theorem}{Theorem}[section]
\newtheorem{assumption}[theorem]{Assumption}
\crefname{assumption}{Assumption}{Assumptions}
\theoremstyle{definition} 
\newtheorem{definition}[theorem]{Definition}
\newtheorem{remark}[theorem]{Remark}
\usepackage{mathtools}

\usepackage{mdframed}
\surroundwithmdframed[
topline=false,
rightline=false,
bottomline=false,
leftmargin=\parindent,
skipabove=\medskipamount,
skipbelow=\medskipamount,
linewidth=1.5pt,
linecolor=black!25
]{proof}

\usepackage{xspace}

\newcommand{\coloneq}{\coloneqq}
\newcommand{\eqcolon}{\eqqcolon}

\newcommand{\diag}{\mathop{\mathrm{diag}}\nolimits}

\usepackage[style=alphabetic,sorting=nyt,hyperref,labelalpha,maxnames=100,mincitenames=4,maxcitenames=5]{biblatex}

\newcommand{\mat}[1]{\left\lbrack\begin{matrix} #1 \end{matrix}\right\rbrack}

\newcommand{\vecSmallT}[1]{\lbrack\begin{matrix} #1 \end{matrix}\rbrack^{\transp}}

\newcommand{\transp}{\top}

\newcommand{\R}{\mathbb{R}}
\newcommand{\Z}{\mathbb{Z}}
\newcommand{\N}{\mathbb{N}}

\newcommand{\ie}{i.\,e.\xspace}
\newcommand{\eg}{e.\,g.\xspace}

\makeatletter
\newcases{lrdcases}
{\quad}
{$\m@th\displaystyle{##}$\hfil}
{$\m@th\displaystyle{##}$\hfil}
{\lbrace}
{\rbrace}
\newcases{lrdcases*}
{\quad}
{$\m@th\displaystyle{##}$\hfil}
{{##}\hfil}
{\lbrace}
{\rbrace}
\makeatother

\newcommand{\smashUnderbrace}[2]{\vphantom{#1}\smash{\underbrace{#1}_{#2}}}

\begin{document}

\title{Convergence Rate Abstractions for Weakly-Hard Real-Time Control}
\subtitle{A General Framework for Analysis and Co-Design of Control and Scheduling}
\author{Maximilian Gaukler, Tim Rheinfels, Peter Ulbrich, Günter Roppenecker\\
Friedrich-Alexander-Universität Erlangen-Nürnberg\\
\small max.gaukler@fau.de, tim.rheinfels@fau.de, peter.ulbrich@fau.de, guenter.roppenecker@fau.de
}
\setkomafont{paragraph}{\bfseries \rmfamily}
\setkomafont{author}{\normalsize}
\setkomafont{date}{\normalsize}
\subject{Technical Report}
\date{first published Dec 2019; revised June 2020 (extended \cref{sec:related-work})}

\maketitle

\paragraph{Abstract}
Classically, a control loop is designed to be executed strictly periodically. This is, however, difficult to achieve in many scenarios, for example, when overload or packet loss cannot be entirely avoided. Here, \emph{weakly-hard} real-time control systems are a common approach which relaxes timing constraints and leverages the inherent robustness of controllers.
Yet, their analysis is often hampered by the complexity arising from the system dimension and the vast number of possible timing sequences.
In this paper, we present the novel concept of \emph{convergence rate abstractions} that provide a sound yet simple one-dimensional system description. This approach simplifies the stability analysis of weakly-hard real-time control systems. At the same time, our abstractions facilitate efficient computation of bounds on the worst-case system state at run-time and thus the implementation of adaptation mechanisms.

\section{Motivation}

Traditionally, control systems and their real-time execution platforms are designed as independent entities. This split approach requires deterministic scheduling and execution of the real-time tasks implementing the controllers. Consequently, the systems must be designed for the worst-case scenario of maximum disturbance and execution times, which implies excessive over-provisioning of computing resources. Given the ever-increasing application and system complexity alongside the pressure to utilize powerful but non-deterministic general-purpose hardware, the inevitable pessimistic overapproximations become unacceptable.

Another approach that is becoming increasingly accepted is, therefore, to co-design controller and real-time system. First introduced by Seto et al. \cite{seto:96:rtss}, co-design aims to (1) relax scheduling requirements, that is periodicity and deadline adherence, and to (2) adapt the controller design while still guaranteeing stability. The approaches developed since range from more flexible task models~\cite{dai:19:emsoft} and the adjustment of sampling periods~\cite{greco:11:toac,castane:06:ecrts,cervin:02:rts} to weakly-hard real-time systems that shift away from stringent deadlines. The latter gained wide popularity in the form of $(m,K)$-firm scheduling \cite{Hamdaoui1995,ramanathan:97:isftc}: Instead of executing every job instance of a task, only at least $m$ out of $K$ consecutive job releases must meet their deadlines.

The downside of easing real-time scheduling demands, however, is to maintain the control stability. Two approaches can be distinguished: (1) To limit the adaptation potential of the real-time scheduling such that even in the worst execution scenario, the controller can handle the worst possible disturbance. This corresponds to a static design that can only benefit from the robustness inherent to the controller. (2) Dynamic scheduling that optimizes resource allocation under a given maximum error constraint. Here, information on the current system state is leveraged to extend the adaptation potential in benign disturbance scenarios.

From a design point of view, adaptation at run-time is undoubtedly beneficial. However, to ensure adequate scheduling of jobs, dynamic estimation of control stability is necessary at run-time as well. In practice, this analysis entails a trade-off between exact models, which are often prohibitively complex to compute, and simple models that inevitably impose pessimism.

\begin{figure}
\begin{centering}
	\includegraphics{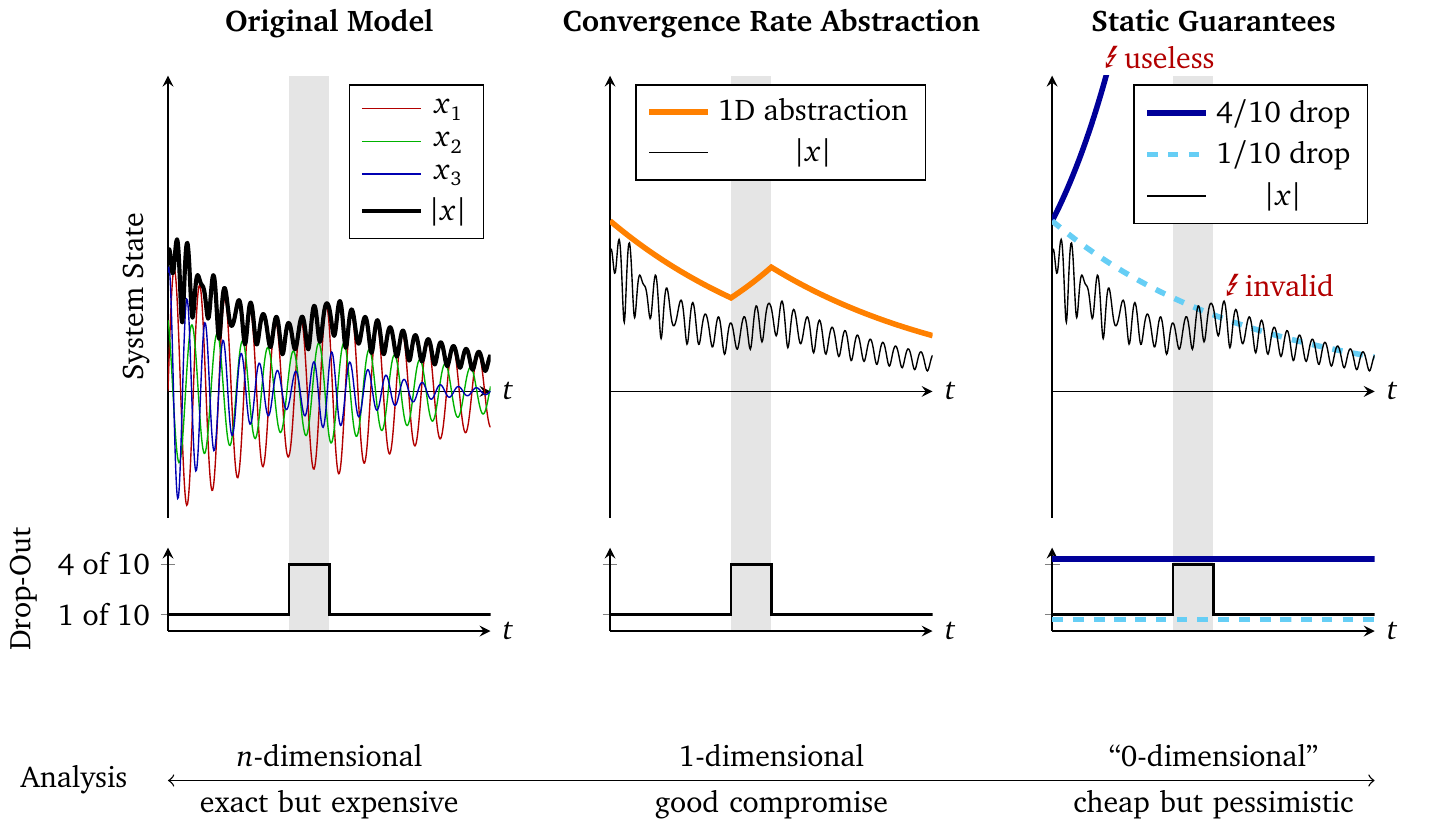}
\end{centering}
\caption{This illustration depicts that dynamic convergence rate abstractions (center) aim to hit a sweet spot for stability analysis between the high complexity of the original system~(left) and the pessimism resulting from static guarantees~(right).}
\label{fig:intro}
\end{figure}

\subsection{Problem Illustration}

We showcase the resulting dilemma by the exemplary system in \cref{fig:intro}, which is defined by state $x=\vecSmallT{x_1 & x_2 & x_3}$ with magnitude $|x|$ and disturbed by a known but varying drop-out of actuations (e.g., deadline misses, packet losses). We assume the system to be stable for a 1-of-10 drop-out with rare short-term episodes of 4-of-10 (transient overload).
Analysis of the full-dimensional system model (\cref{fig:intro}, left) allows for exact results: the system remains stable, and the initial disturbance decays despite the short-term overload situation.
However, computing state bounds from the full-dimensional model is associated with extensive overheads, especially because all possible combinations of timing and disturbance must be considered. This impedes evaluation as part of run-time scheduling.

Contrary, static analysis (\cref{fig:intro}, right) is done at design-time and is thus relatively simple and free of run-time overheads. However, it is inherently plagued by analysis pessimism as it must always assume the worst execution pattern. In our example, the system state is unstable (exponentially growing) for the assumed maximum 4-of-10 drop-out, rendering static 4-of-10 guarantees \emph{useless}. Therefore, the considered scenario cannot be shown stable using static guarantees within a \enquote{$n$-of-10} framework. At the same time, analysis for a 1-of-10 drop-out (dotted line) is quite close to reality, yet its result becomes \emph{invalid} in 4-of-10 conditions due to the violated assumption of a 1-of-10 drop-out.

In this paper, we aim to solve the aforementioned dilemma between sound state estimation and run-time efficiency. We, therefore, introduce the concept of one-dimensional convergence rate abstractions (\cref{fig:intro},~center) that hit the sweet spot between complexity and pessimism: While the bound for $|x|$ is somewhat pessimistic, stability can be shown without considering the whole complexity of the original system. The \emph{statefulness} of our abstractions enables dynamic yet efficient estimation: In simple terms, the abstraction computes a scalar \emph{damage counter} variable that increases or decreases over time depending on whether the current drop-out is large or small. By that context-sensitivity, the convergence rate abstraction can fully exploit the transient nature of the assumed overload condition.

\subsection{Solution Statement}
\label{sec:intro:solution}

Formally, for the special case of $(m,K)$-firm scheduling with \emph{fixed} parameters, there is an existing approach that implicitly uses a one-dimensional abstraction for the evaluation of stability: In \cite{Huang2019}, the authors prove the stability of a nonlinear sampled-data control loop by considering an upper bound for $|x|$ similar to the graph in \cref{fig:intro} (center). Remarkably, their approach works without constructing a Lyapunov function or a solution of the dynamics. Instead, there are only two relevant requirements, which correspond to the decaying and increasing phases of the graph: (1) Exponential stability in the nominal case, that is a guarantee on the decay rate if the controller is never skipped. (2) A Lipschitz bound on the dynamics, which effectively bounds the growth rate during drop-outs of the controller. Using only these parameters, however, may come at the cost of increased pessimism.

In this paper, we present a framework that generalizes these benefits to time-varying execution conditions: We show that the above requirements yield a \emph{one-dimensional convergence rate abstraction} of the system dynamics. Consequently, the stability of the original system can be shown by proving the stability of the corresponding abstract system. We introduce this concept and its generalization to wider classes of uncertainties, such as disturbance and uncertain input/output timing. Moreover, in this paper, we discuss alternative Lyapunov-based abstractions to reduce analysis pessimism.

\newpage

\tableofcontents

\newpage

\section{Notation}
\begin{itemize}
	\item The notation $a:=b$ denotes that $a$ is defined as equal to $b$.
	
	\item $\R$ is the set of real numbers, $\N:=\{1,2,\dots\}$ the natural numbers, $\N_0:=\N \cup \{0\}$, and $\Z := \{\dots,-2,-1,0,1,2,\dots\}$ the set of integers.
	
	\item If not stated differently, $k$ refers to any $k \in \N_0$, and $x$ to $x \in \R^n$.
	
	\item The abbreviation $a_k \equiv 0$ denotes $a_0=a_1=\dots=0$.
	
	\item For $x \in \R^n, A \in \R^{n \times n}$, $x^\transp$ denotes the transpose, $|x|:=\sqrt{x^\transp x}$, and
	\begin{equation}
		\|A\|_2:=\max_{x \in \R^{n}\setminus\{0\}} \frac{|Ax|}{|x|}
	\end{equation}
	 denotes the spectral norm.
	
	\item $\mathcal B_v \coloneq \lbrace x \in \R^{n} | ~ |x| \le v \rbrace$ is the unit ball of radius $v$.
	
	\item $\lambda_i\{A\}$, $i=1,\dots,n$ are the eigenvalues of $A$ in arbitrary order.
	
	\item Rounding towards negative infinity (floor function) is defined as
	\begin{equation}
	\lfloor x\rfloor := \max \{ n \in \Z ~|~ n \leq x \},
	\end{equation}
	\eg, $\lfloor 1.9 \rfloor = \lfloor 1 \rfloor = 1$.
	The modulo operation is thereby defined as
	\begin{equation}
		x \bmod y := x - \left \lfloor \frac{x}{y} \right \rfloor y.
	\end{equation}
\end{itemize}

\section{Problem Statement for Weak Execution} \label{sec:abstraction:problem}
In this section, we formalize the problem of control stability under weakly-hard execution and introduce the concept of convergence rate abstractions by a simplified example.

\subsection{Problem Setting}
\paragraph{Given} The states of plant and controller are combined into one state vector $x_k$. Then, the closed loop dynamics are given as
\begin{equation}
	x_{k+1}=A_{\sigma_k}(x_k) + w_k, \quad \sigma_k \in \Sigma, \quad x_0 \in \mathbb R^{n}, \label{eq:abstraction-orig-system}
\end{equation}
where $\sigma_k=0$ means that the controller is executed normally in the $k$-th step and $\sigma_k \neq 0$ means the execution differs from the normal case. $\Sigma \supseteq \{0\}$ is the set of possible execution modes. For example, $\Sigma=\{0,1\}$ describes the simple case where the controller is either completely skipped ($\sigma=1$) or executed ($\sigma=0$). The disturbance $w_k$ is present mainly for the derivations; it will be assumed as zero in most of the resulting stability criteria.

In the general nonlinear case, $A_{\sigma}(x)=f(\sigma, x)$ is a function. For the linear case, we slightly abuse notation and define it as the matrix multiplication $A_{\sigma}(x)\coloneqq A_{\sigma} \cdot x$ with $A_{\sigma} \in \R^{n \times n}$, which means that the braces may be omitted.

Two assumptions are used throughout this section: Lipschitz continuity and nominal exponential stability, as defined in the following.

\begin{assumption}[Lipschitz Continuity] \label{assumption:lipschitz}
	 $A_\sigma(x)$ is continuously differentiable and Lipschitz continuous in $x$, \ie, there is a constant $L \geq 0$ such that
	\begin{equation}
	|A_{\sigma}(x + \Delta x) - A_{\sigma}(x)| \le L |\Delta x| \quad \forall x, \Delta x \in \R^n, \forall \sigma \in \Sigma. \label{eq:lipschitz}
	\end{equation}
	For the linear case, $L=\max_{\sigma \in \Sigma} \|A_{\sigma}\|_2$.
\end{assumption}

\begin{assumption}[Nominal Exponential Stability] \label{assumption:nominal-exp}
For the dynamics \eqref{eq:abstraction-orig-system}, there exist an overshoot factor $\alpha\geq1$ and a growth rate $\rho \in ( 0 ; 1)$ such that
\begin{equation}
	(\forall k \in \N_0: \quad \sigma_k =0 \land w_k=0) \quad \Rightarrow \quad \left( \forall x_0 \in \R, k\in \N_0:\quad |x_k| \le \alpha \rho^k |x_0| \right). \label{eq:abstraction-nominal-exp-stability}
\end{equation}
\end{assumption}

\subsection{Goal}
The goal is to reduce the closed-loop dynamics to a one-dimensional \emph{abstract system}, \ie,
\begin{equation}
	v_{k+1} = \rho_{\sigma_k} v_k + \beta |w_k|,
\end{equation}
with the one-dimensional state $v_k \in [0,\infty)$ that satisfies the \emph{abstraction guarantee}
	\begin{equation}
	|x_k| \le v_k \quad \forall k \geq 0
	\end{equation}
and describes the original system well enough to show the desired stability and quality (cf.~\cref{fig:intro} on p.~\pageref{fig:intro}). Before we formalize this concept in \cref{sec:formalization}, we will first discuss how it naturally arises from exponential stability.

\subsection{Exponential Stability as a Simple Abstraction}

Exponential stability can be seen as an abstraction. First, we will revisit the derivation of exponential stability for linear time-invariant systems.

\begin{remark}[Exponential Stability for Linear Time-Invariant Systems]
To determine a valid combination of $\rho$ and $\alpha$ fulfilling \eqref{eq:abstraction-nominal-exp-stability}, first choose any $\rho$ such that
\begin{equation}
	\rho > \max_i |\lambda_i\{A_0\}| \quad \land \quad \rho<1. \label{eq:exp-stable-linear:rho}
\end{equation}
Then, as we will shown soon, exponential stability \eqref{eq:abstraction-nominal-exp-stability} holds if and only if
\begin{equation}
	\alpha \geq \alpha_{\min} := \max_{k \in \{0,1,\dots,\tilde k-1\}} \left\|A_0^k \rho^{-k}\right\|_2,
	\label{eq:exp-stable-linear:alpha-first}
\end{equation}
where the above maximum needs to be evaluated up to the finite number
\begin{equation}
\tilde k := \min \left \{ k \in \N \Big|  \left\|A_0^{ k} \rho^{- k}\right\|_2 < 1 \right \} < \infty. \label{eq:exp-stable-linear:k-tilde}
\end{equation}
Both $\tilde k$ and $\alpha_{\min}$ can therefore be evaluated numerically.
\end{remark}
\begin{proof}~
\paragraph*{Value of $\rho$:} While the upper bound of \eqref{eq:exp-stable-linear:rho} is obvious, a detailed proof for the lower bound would be beyond the scope of this paper. It is closely linked to the definition and properties of the spectral radius \cite{Jungers2009}.

	Note that the edge case of $\rho = \max_i |\lambda_i\{A_0\}|$ is possible for some systems, \eg, the one-dimensional case $A_0=\rho$, but not in general: For example, let
\begin{equation}
A_0=\mat{0 & 0\\ 1 & 0}, ~x_0=\mat{1\\0}.
\end{equation}
Here it is impossible to choose $\rho=\max_i |\lambda_i\{A_0\}|=0$ because $|x_1|=1 \not \leq \alpha \rho^k |x_0| = 0$. This is closely related to the question of Reducibility and the existence of Extreme Matrix Norms discussed in \cite{Jungers2009}. 

\paragraph*{Value of $\alpha$:} The range of $\alpha$ satisfying exponential stability \eqref{eq:abstraction-nominal-exp-stability} with the chosen value of $\rho$ can then determined by rewriting \eqref{eq:abstraction-nominal-exp-stability}. In the following, $\forall x_0,k$ is shorthand for $\forall x_0 \in \R, k\in \N_0$.
\allowdisplaybreaks
\begin{align}
	\text{\eqref{eq:abstraction-nominal-exp-stability}} \Leftrightarrow & \left( \forall x_0, k:\quad |A_0^kx_0| \le \alpha \rho^k |x_0| \right) \\
	\Leftrightarrow&
	\left( \forall x_0, k:\quad \begin{lrdcases}
	0 \le 0, & x_0 = 0\\
	\displaystyle \frac{|A_0^kx_0|}{\rho^k |x_0|} \le \alpha , & x_0 \neq 0
	\end{lrdcases} \right) \\
	\stackrel{\rho>0}{\Leftrightarrow}& \left( \forall x_0 \neq 0, k:\quad \alpha \ge \frac{|A_0^k \rho^{-k} x_0|}{|x_0|} \right) \label{eq:exp-stable-linear:alpha-temp}\\
	\Leftrightarrow& \alpha \stackrel{(*)}{\geq} \max_{k \in \N_0}\left(\max_{x_0 \in \R^{n}\setminus\{0\}} \frac{\left|A_0^k \rho^{-k} x_0\right|}{|x_0|}\right) = \max_{k \in \N_0} \left\|A_0^k \rho^{-k}\right\|_2 \stackrel{(**)}{=} \underbrace{\max_{k \in \{0,1,\dots,\tilde k-1\}} \left\|A_0^k \rho^{-k}\right\|_2}_{\alpha_{\min}}
\end{align}

The statement $(**)$ will be proven later. In $(*)$ we silently assume that the maximum over the considered infinite sets exists, \ie, it is neither infinite nor merely a supremum. The validity of this assumption with regard to $x_0$ becomes clear from the definition of the spectral norm as the maximum (not supremum)
\begin{equation}
\|A\|_2 = \max_{x \in \R^{n}\setminus\{0\}} \frac{|Ax|}{|x|}.
\end{equation} Regarding $k$ it can be seen from $(**)$, which results in the maximum over a \emph{finite} set.

\paragraph{Existence of $\tilde k < \infty$:}
A finite $\tilde k$ exists because
\begin{equation}
	\lim_{k\to\infty} A_0^k \rho^{- k} =\lim_{k\to\infty} \big(\underbrace{A_0 \rho^{-1}}_{\mathrlap{\hspace{-1.4em}|\lambda_i\{\cdot\}| = |\rho^{-1} \lambda_i\{A_0\}| \stackrel{\text{\eqref{eq:exp-stable-linear:rho}}}{<} 1 \quad \forall i }}\big)^k = 0
\end{equation}
and therefore, by the epsilon-delta-definition of a limit, for any $\epsilon > 0$ there is a corresponding $\delta(\epsilon)<\infty$ such that
\begin{equation}
	\|A_0^k \rho^{- k}\|_2 < \epsilon \quad \forall k \geq \delta(\epsilon). \label{eq:exp-stable-linear:k-infty-lim-delta}
\end{equation}
Choose $\epsilon=1$ to see that by \cref{eq:exp-stable-linear:k-tilde,eq:exp-stable-linear:k-infty-lim-delta} there exists $\tilde k \leq \delta(1) < \infty$.

\paragraph*{Value of $\tilde k$:}
Consider an arbitrary $k \geq 0$. Let $n\geq 0$ be such that $n\tilde  k \leq k < (n+1) \tilde k$.  By the submultiplicativity $\|AB\|_2 \leq \|A\|_2\|B\|_2$, we can split off $n$ factors $\|A_0^{\tilde k} \rho^{\tilde k}\|_2 < 1$ from the term $\|A_0^{k} \rho^{k}\|_2$:
	\begin{align}
\left\|A_0^k \rho^{-k}\right\|_2 \leq & \left\|A_0^{\tilde k} \rho^{-\tilde k}\right\|_2^n \left\|A_0^{k - n\tilde k} \rho^{-(k - n\tilde k)}\right\|_2\\ \stackrel{\text{\eqref{eq:exp-stable-linear:k-tilde}}}{\leq} & \left\|A_0^{k - n\tilde k} \rho^{-(k - n\tilde k)}\right\|_2\\
  \stackrel{\mathclap{(***)}}{\leq} &\max_{k \in \{0,1,\dots,\tilde k-1\}} \left\|A_0^k \rho^{-k}\right\|_2, \label{eq:exp-stable-linear:k-tilde-tmp}
	\end{align}
		where $(***)$ holds due to $0 \leq k-n\tilde k < \tilde k$.
This proves $(**)$:
\begin{align}
		\max_{k \in \N_0} \left\|A_0^k \rho^{-k}\right\|_2 =
		& \max\left\lbrace \max_{k \in \{0,1,\dots,\tilde k-1\}} \left\|A_0^k \rho^{-k}\right\|_2, ~ \max_{k \in \{\tilde k, \tilde k+1, \dots\}} \left\|A_0^k \rho^{-k}\right\|_2 \right\rbrace\\\stackrel{
		\mathclap{\text{\eqref{eq:exp-stable-linear:k-tilde-tmp}}}}{=}&
		\max_{k \in \{0,1,\dots,\tilde k-1\}} \left\|A_0^k \rho^{-k}\right\|_2.
\end{align}
\end{proof}

The analysis of exponential stability leads to a first abstraction:

\begin{theorem}[Exponential Decay of Disturbance in the Nominal Case] \label{thm:exp-decay}
Let $\sigma_i \equiv 0$ and \cref{assumption:lipschitz,assumption:nominal-exp} hold. Then, there exists a constant $\beta \geq 1$ such that the resulting system
\begin{equation}
	x_{k+1} = A_0(x_k) + w_k \label{eq:abstraction-reformulated}
\end{equation}
obeys
\begin{equation}
	|x_{k}| \le v_{k} \forall k \label{eq:abstraction-guarantee}
\end{equation}
with
\begin{equation}
	v_{k+1} = \rho v_k + \beta |w_k|, \quad v_0 = \alpha |x_0|. \label{eq:abstraction-dynamics}
\end{equation}

This is a simple example for a one-dimensional \emph{convergence rate abstraction}: The one-dimensional $v$-system summarizes the relevant information about stability of the $n$-dimensional $x$-system. A general definition of this concept will be given later.
\end{theorem}

\begin{proof}[Proof for Linear Systems]

Unrolling the recursion \eqref{eq:abstraction-reformulated} yields
\begin{align}
x_k &= A_0^{k} x_0  + \sum_{i=0}^{k-1} A_0^{i} w_{k-1-i}\\
\Rightarrow |x_k| &\le |A_0^k x_0 | +  \sum_{i=0}^{k-1} |A_0^{i} w_{k-1-i}|. \label{eq:abstraction-reformulated-proof-lin-1}
\end{align}
It should be noted that this step explicitly requires the superposition property and is not directly applicable to nonlinear systems. For linear systems, the exponential stability
\eqref{eq:abstraction-nominal-exp-stability} can be concretized to
\begin{equation}
\forall x \in \R^n, i \in \N_0:\quad  |A_0^{i} x| \le \alpha \rho^i |x|.
\end{equation}
Applying this to each summand of \eqref{eq:abstraction-reformulated-proof-lin-1} yields
\begin{align}
|x_k| &\le \alpha \rho^k |x_0| +  \sum_{i=0}^{k-1} \alpha \rho^i |w_{k-1-i}|. \label{eq:abstraction-exp-linear-bound}
\end{align}
Explicitly solving the abstraction's dynamics \eqref{eq:abstraction-dynamics} yields
\begin{align}
v_k &= \alpha \rho^k |x_0| + \sum_{i=0}^{k-1} \beta \rho^i |w_{k-1-i}|,
\end{align}
which shows that $|x_k| \le v_k$ for any choice of $\beta \geq \alpha$.

\end{proof}

\begin{proof}[Proof Sketch for Nonlinear Systems]
For $\sigma_i \equiv 0$, the system can be written as
\begin{equation}
x_{k+1}=f(x_k) + w_k \quad \text{with} \quad f(x) \coloneqq A_0(x)
\end{equation}
Due to assumption \eqref{eq:abstraction-nominal-exp-stability}, it is exponentially stable with coefficients $\rho$ and $\alpha$.

(\emph{Note:} The following step is only a sketch and not yet a rigid proof.)
Using a converse Lyapunov theorem similar to \cite[Theorem 2.7]{Bof2018}\footnote{The cited theorem only considers a bounded set for $x$, whereas we consider unbounded $x$ to simplify the notation. However, the argument should still hold due to our explicit requirement of Lipschitz continuity. Additionally, if stability is shown, then the set is guaranteed to be bounded.}, it can \emph{probably} be shown that exponential stability with parameters $\rho$ and $\alpha$, Lipschitz continuity and  continuous differentiability are sufficient for the existence of a \enquote{square-like} Lyapunov function $\tilde V(x)$ and a constant $\beta \geq \alpha$ such that
\begin{subequations}
	\begin{align}
	\forall x:&& |x|^2 \leq \tilde V(x) \leq& \alpha^2 |x|^2\\
	\forall x:&& \tilde V(f(x)) \leq& \rho^2 \tilde V(x) \\
	\forall x,\Delta x:&& \tilde V(x+\Delta x) \leq&  \tilde V(x) + \beta^2 |\Delta x|^2,
	\end{align}
\end{subequations}
where $\beta \geq \alpha$ can be derived from the Lipschitz bound on $f(x_k)$.

Taking the square root of the above equations and applying
\begin{equation}
	\sqrt{a+b} \le \sqrt{a} + \sqrt{b} \quad \forall a,b \ge 0
\end{equation}
shows that there exists a \enquote{linear-like} Lyapunov function $V(x)=\sqrt{\tilde V(x)}$ such that
\begin{subequations}
	\begin{align}
	\forall x:&& \quad |x| \leq V(x) \leq& \alpha |x| &&\text{(linearly bounded)} \label{eq:abstraction-nonlin-bound} \\
	\forall x:&& \quad V(f(x)) \leq& \rho V(x) &&\text{(exponentially converging)} \label{eq:abstraction-nonlin-rho} \\
	\forall x, \Delta x:&& \quad V(x+\Delta x) \leq&  V(x) + \beta |\Delta x| &&\text{(globally Lipschitz continuous)} \label{eq:abstraction-nonlin-delta}
	\end{align}
\end{subequations}

The criteria \labelcref{eq:abstraction-nonlin-bound,eq:abstraction-nonlin-delta,eq:abstraction-nonlin-rho} lead to
\begin{align}
V(x_{k+1}) = V(f(x_k) + w_k) \stackrel{\text{\eqref{eq:abstraction-nonlin-delta}}}{\leq} \beta |w_k| + V(f(x_k)) \stackrel{\text{\eqref{eq:abstraction-nonlin-rho}}}{\leq} \beta |w_k| + \rho V(x_k). \label{eq:abstraction-nonlin-lyap-dynamics}
\end{align}
The dynamics  \eqref{eq:abstraction-dynamics} of the abstraction $v_k$ now track the upper bound for $V(x_k)$ arising from  $V(x_0) \stackrel{\text{\eqref{eq:abstraction-nonlin-bound}}}{\leq} \alpha |x_0|$ and \eqref{eq:abstraction-nonlin-lyap-dynamics}. Therefore,
\begin{equation}
V(x_k) \stackrel{\text{\eqref{eq:abstraction-nonlin-bound} $\land$ \eqref{eq:abstraction-nonlin-lyap-dynamics} $\land$ \eqref{eq:abstraction-dynamics}}}{\leq} v_k \label{eq:abstraction-nonlin-lyap-bound-vk}
\end{equation}
which ensures the abstraction guarantee \eqref{eq:abstraction-guarantee} by
\begin{equation}
|x_k|  \stackrel{\text{\eqref{eq:abstraction-nonlin-bound}}}{\leq} V(x_k) \stackrel{\text{\eqref{eq:abstraction-nonlin-lyap-bound-vk}}}{\leq} v_k \quad \Rightarrow \quad \text{\eqref{eq:abstraction-guarantee}}.
\end{equation}

\end{proof}

\begin{remark}[Lyapunov Function for a Linear System]
	\label{remark:lyapunov-linear}
	To give a simple example for the Lyapunov function $V(x)$ in the above proof for nonlinear systems, let $f(x)=Ax$. Then, there always exists a Lyapunov function in quadratic form $\tilde V(x)=x^{\transp}Px$, for which the square root $V(x)=\sqrt{x^{\transp}Px}$ fulfills the criteria \labelcref{eq:abstraction-nonlin-bound,eq:abstraction-nonlin-delta,eq:abstraction-nonlin-rho} leading to the abstraction $v_k \geq V(x)$. $P$ is the solution of the Lyapunov equation
	\begin{equation}
	A^\transp P A - P = -Q
	\end{equation}
	with the positive definite parameter $Q$. (Note that in the corresponding MATLAB command \texttt{dlyap}, $A$ is transposed.) As $P$ is positive definite, the Cholesky decomposition $P=R^{\transp}R$ leads to the equivalence $V(x)=\sqrt{x^{\transp}R^{\transp}Rx} = |Rx|$. The transformation $\tilde x \coloneqq Rx$ then describes new coordinates in which the system is contractive ($|\tilde x_{k+1}| \le |\tilde x_k|$ if $w_k=0$). \Cref{fig:contractive_system} provides an example for a dampened harmonic oscillator.
\end{remark}

\begin{figure}[t]
\begin{centering}
	\input{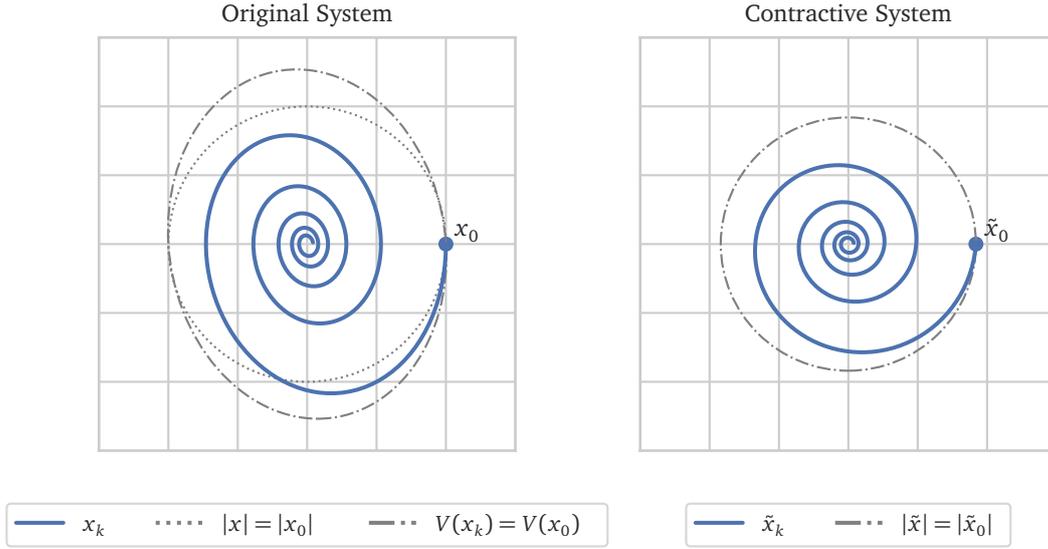}
\end{centering}
\caption{Comparision of the original trajectory $x_k$ (left) and the transformed contractive one $\tilde x_k$ (right) for a stable harmonic oscillator. While the Euclidean norm $|x_k|$ overshoots the initial norm $|x_0|$, the system is contractive in $|\tilde x_k|$. }
\label{fig:contractive_system}
\end{figure}

\subsection{Existence and Pessimism of Simple Abstractions}
For the variant of abstraction provided by \cref{thm:exp-decay}, it is possible to show that it tightly matches stability of the nominal case, but may be pessimistic for weakly-hard execution:

\begin{theorem}[Existence of Convergence Rate Abstractions] Under \cref{assumption:lipschitz} (Lipschitz continuity), the exponential stability of the nominal system (\cref{assumption:nominal-exp} with $\rho<1$) is equivalent to the existence of a stable convergence rate abstraction of the form \eqref{eq:abstraction-guarantee} -- \eqref{eq:abstraction-dynamics} with the same convergence rate $\rho$.
\end{theorem}
\begin{proof}
\emph{\enquote{Exponentially stable \eqref{eq:abstraction-nominal-exp-stability} $\land$ Lipschitz $\Rightarrow$ convergence rate abstraction}:} This is implied by \cref{thm:exp-decay}.

\emph{\enquote{Convergence rate abstraction $\Rightarrow$ exponentially stable \eqref{eq:abstraction-nominal-exp-stability}}:} Assume $w_k \equiv 0$ and $\sigma_k \equiv 0$. The result follows from unrolling the recursion of the abstraction \eqref{eq:abstraction-guarantee} -- \eqref{eq:abstraction-dynamics}:

\begin{equation}
	 |x_k| \stackrel{\text{\eqref{eq:abstraction-guarantee}}}{\leq} v_k \stackrel{\text{\eqref{eq:abstraction-dynamics}},~ w_k \equiv 0}{=} \rho^{k} v_0 \stackrel{\text{\eqref{eq:abstraction-dynamics}}}{=} \rho^{k} \alpha |x_0|
\end{equation}
\end{proof}

\begin{remark}[Conservatism of Convergence Rate Abstractions under Disturbance]
	\label{remark:abstraction-weakly-hard:conservative}
	As shown by the previous theorem, exponential stability of the nominal case is captured \emph{exactly} by an abstraction, \ie, $\rho$ has the same value in the exponential stability \eqref{eq:abstraction-nominal-exp-stability} as  in the abstraction \eqref{eq:abstraction-dynamics}. However, \emph{for deviations from the nominal case ($w_k \neq 0$), this abstraction  \eqref{eq:abstraction-dynamics} is conservative} in general, as the $n$-dimensional state space of the original $x$-system cannot generally be embedded into the one-dimensional state space of the abstracted $v$-system. These deviations $w_k$ include physical disturbance and weakly-hard execution.

	A concrete example proving this statement will be given later.
\end{remark}

\section{Formal Definition of Convergence Rate Abstractions}\label{sec:formalization}
The structure of the abstraction given by \eqref{eq:abstraction-dynamics} can be generalized as follows:
\begin{definition}[One-Dimensional Convergence Rate Abstraction] \label{def:abstraction}
Let
\begin{equation}
	x_{k+1} = f(x_k, w_k), \quad k \in \N_0 \label{eq:def:growth-bound:orig}
\end{equation}
be the original system with state $x_k \in \R^n$ and partially unknown disturbance $w_k \in \mathcal{W} \subseteq \R^{n_w}$. Then, the one-dimensional system
\begin{equation}
v_{k+1} = \bar f(v_k, \bar w_k) \qquad  \text{(\eg, }v_{k+1} = \rho v_k + \beta \bar w_k\text{)} \label{eq:def:growth-bound:orig:dyn}
\end{equation}
with state $v_k \in [0, \infty)$ and incomplete disturbance information
\begin{equation}
	\bar w_k \in \bar G(w_k) \subseteq [0, \infty) \qquad  \text{(\eg, }\bar w_k \geq |w_k|\text{\,)} \label{eq:def:growth-bound:orig:w}
\end{equation}
 is a \emph{convergence rate abstraction} for \eqref{eq:def:growth-bound:orig} iff
\begin{itemize}
	\item the initialization guarantees
	\begin{equation}
		v_0 \ge \bar h(x_0) \qquad \text{(\eg, }v_0 \geq \alpha |x_0|\text{)}\label{eq:def:growth-bound:orig:init-guarantee}
	\end{equation}
	\item and the functions $\bar f, \bar G, \bar h$ guarantee
	\begin{equation}
	\forall x_0: \quad \left(v_0 \ge \bar h(x_0) \quad \Rightarrow \quad \left(\forall k \geq 0:~ |x_{k}| \le v_{k} \right) \right). \label{eq:def:growth-bound:orig:guarantee}
	\end{equation}
\end{itemize}

Note that the disturbance $w_k$ in this formulation does not only include physical disturbance but also the weakly-hard execution $\sigma_k$ and any other uncontrolled time-varying influence.
\end{definition}

The abstraction guarantee \eqref{eq:def:growth-bound:orig:guarantee} enables translating stability analysis of the abstraction to the original system:
\begin{theorem}
	If a system's convergence rate abstraction as per \cref{def:abstraction} is exponentially stable (\,$|v_k| < \alpha \rho^k v_0$\,) or has bounded state ($|v_k|<\text{const}$), the same holds for the original system.
\end{theorem}
\begin{proof}[Proof Sketch]
	This immediately follows from \labelcref{eq:def:growth-bound:orig:guarantee,eq:def:growth-bound:orig:init-guarantee}, as then $0 \le |x_k| \le v_k$.
\end{proof}

\paragraph{Connection to abstractions and reachable sets} To explain why the above concept is called an abstraction, we will first define this term. In informal words, abstraction means mapping the representation of a problem into a new representation that preserves desirable properties and is easier to handle \cite{Giunchiglia1992}.
Here, the original $x$-system is mapped to the abstract $v$-system, which preserves the upper bound $|x_k|<v_k$ and is easier to analyze because it is one-dimensional.

For the scope of this publication, a more concrete definition of an abstraction can be given using the view of Behavioral Systems Theory \cite[Chapter 1]{willems2013introduction}, where a system is described by the set of possible combinations of input and output trajectories.
From this point of view, the original system \eqref{eq:def:growth-bound:orig} is the set
\begin{equation}
	S := \left\{\left((x_0, x_1,\dots), (w_0,w_1,\dots)  \right)  ~\Big|~ x_i \in \R^n, w_i \in \mathcal{W}, \text{\eqref{eq:def:growth-bound:orig}} \right\}.
\end{equation}
We define that a system $\tilde S$ is an \emph{abstraction} of $S$ if and only if $\tilde S \supseteq S$, \ie, the abstracted system contains at least the original trajectories and possibly others.

Based on the relation $x_k \in \mathcal{B}_{v_k}$, the convergence rate abstraction $\bar f, \bar G$ can be written as the system
\begin{align}
	\tilde S := \Big\{\left((x_0, x_1,\dots), (w_0,w_1,\dots)  \right)  ~\big|~& x_k \in \mathcal{B}_{v_k} \text{with $\nu_k$ subject to \labelcref{eq:def:growth-bound:orig:init-guarantee,eq:def:growth-bound:orig:dyn,eq:def:growth-bound:orig:w}},& \nonumber\\
	& w_k \in \mathcal{W}, \forall k \in \N_0 &\Big\}.
\end{align}
This new system $\tilde S$ is an abstraction of $S$ according to the above definition, as the set of possible $\tilde x_k$ always contains the actual trajectory $x_k$. Equivalently, the ball $\mathcal{B}_{v_k}$ is an outer approximation of the reachable set of \eqref{eq:def:growth-bound:orig}.

\paragraph{Multi-Dimensional Generalization} The abstract state $v_k$ could be higher-dimensional if a general relation $x_k \in \bar S(v_k) \subseteq \R^{n}$ is used instead of  $|x_k|\le v_k$. For example, the $m$ components of $v_k$ could represent the $m$ slowest eigenmovements of a system. Constructing this $v$-system is a question of model order reduction with dynamic error bounds.

\section{Related Work and Contribution}
\label{sec:related-work}
The idea of abstracting a system to an upper bound of its state radius is implicitly contained in the notion of exponential stability, so it has existed in hidden form for centuries. In the following we would like to point out select examples from the literature in which the idea appeared more explicitly.

As discussed in \cref{sec:intro:solution}, a one-dimensional inequality for the state radius is used in \cite{Huang2019} to analyze stability of a nonlinear control loop under weak execution. This reduces the stability test to a system of linear inequalities. The test does not require explicit system dynamics, but only requires bounds on the nominal stability and Lipschitz constant.

A one-dimensional growth and perturbation bound for the open-loop plant is termed \emph{incremental forward completeness} in \cite{Zamani2012}. There, it is used to guarantee the correctness of a discrete-state controller determined from a state space discretization of the plant. In contrast to our work, the convergence rate only abstracts the open-loop plant, so that an unstable plant will result in an unstable convergence rate even if the closed loop is stable. However, it should be possible to transfer the results given in terms of Input-to-State Lyapunov functions to closed-loop analysis. A similar but full-dimensional convergence rate is used in \cite[Section VIII.C]{Reissig2017}.

In \cite[Chapters 6.1, 6.2.2]{Bund2017} networked control systems with packet loss are modeled in an abstract \emph{interval domain} by means of a magnitude impulse response. To some extent this corresponds to the impulse response of our one-dimensional abstraction. However, \cite{Bund2017} employs an explicit formulation of the linear dynamics to generate this impulse response, resulting in a non-exponentially decaying response, which corresponds to an abstract system of order higher than one. The concept is extended to \emph{signal densities}, which represent information on the temporal distribution of signals, \eg, short-term spikes vs. long-term persistent disturbance.

The idea of abstracting the state space to certain sets, especially ellipsoids, is widely used for the formal verification of computer programs and dynamic systems in general \cite{Roux2012}. In some sense, our approach results in a radial abstraction, as it yields spherical sets $|x_k| \le \mathrm{const}$. If a quadratic form is used for the Lyapunov function, as in \cref{remark:lyapunov-linear}, then our approach resembles an ellipsoid abstraction multiplied by a constant (cf. \cref{eq:abstraction-nonlin-bound}) corresponding to the worst-case ratio between the ellipsoid and a sphere. This will be detailed later.

A similar ellipsoid abstraction appears in some variants of Tube Model Predictive Control, \eg, \cite[Section III-B]{Cannon2011}, to describe the tube of possible disturbed trajectories $x(t)$ around the disturbance-free nominal trajectory $\tilde x(t)$. This tube
\begin{equation}
	(x(t) - \tilde x(t))^\transp Q\, (x(t) - \tilde x(t)) \le r^2(t)
\end{equation}
has ellipsoid cross section with fixed shape matrix $Q \in \R^{n \times n}$ but varying size $r(t) \ge 0$. Therefore, the dynamics of $r(t)$ are a one-dimensional abstraction of the influence of disturbance.

More generally, the concept of stability analysis by reduction to a simple dynamical system is formalized by \emph{comparison theory} and the notion of \emph{stability preserving mappings} \cite{Michel2001}. To illustrate the idea in the terms of \cref{def:abstraction}, consider a mapping $V(x)$ from $\R^{n}$ to $[0, \infty)$, similar to a Lyapunov function. We require the bounds
\begin{equation}
	0 \le |x| \le V(x) \le c |x| \qquad \forall x \label{eq:related-work:comparison-state-bound}
\end{equation}
and a worst-case dynamics bound
\begin{equation}
	0 \le V(f(x_k, w_k)) \le \tilde f(V(x_k), w_k) \qquad \forall x_k, w_k. \label{eq:related-work:comparison-bound}
\end{equation}
Then, a one-dimensional \emph{comparison system} is given by the difference inequality
\begin{equation}
	0 \le v_{k+1} \le \tilde f(v_k, w_k), \quad v_0=V(x_0). \label{eq:related-work:comparison-sys}
\end{equation}
This nondeterministic $v_k$-system generates all trajectories for $V(x_k)$ that are possible according to the bound \eqref{eq:related-work:comparison-bound}. Consequently, every trajectory for $V(x_k)$ resulting from the actual $x_k$-dynamics is contained in the set of trajectories of $v_k$. Therefore, if the $v_k$-system is stable (resp. bounded), then $V(x_k)$ converges (is bounded) and by \eqref{eq:related-work:comparison-state-bound} the same holds for $x_k$.

In summary, stability (boundedness) of the comparison system implies stability (boundedness) of the original system. The converse is not generally true because $v_k$ may grow faster than $V(x_k)$ whenever the upper bound of \eqref{eq:related-work:comparison-bound} is pessimistic. These results are generalized in \cite[Proposition 4.1.3]{Michel2001}.

In the special case in which the stability of a comparison system is equivalent to the stability of the original system, $V(x)$ is a \emph{stability preserving mapping}. Convergence rate abstractions are an upper-bound variant of such mappings, as formalized in \cite[Theorem 3.4.1]{Michel2001}. The idea of choosing $V(x)$ based on a Lyapunov function will be revisited later. Note that as discussed in \cref{thm:exp-decay}, it is also possible to obtain an abstraction without explicitly determining $V(x)$. In this case, there is no explicit connection to stability-preserving mappings.

These examples from the literature show that the basic idea of one-dimensional abstractions has existed for a long time and in wide a variety of forms. The contribution intended by this paper is twofold: First, we present a single formalism tailored for weakly-hard execution which embraces these ideas from various fields. Second, this formalism separates the analysis of the system dynamics from the subsequent analysis of weakly-hard stability:

Currently, the method chosen for analyzing the dynamics (\eg, robust exponential stability or Lyapunov function synthesis based on various techniques) is often hard-coded in the weakly-hard stability analysis (\eg, exponential stability under $(m,K)$-execution, maximum state under disturbance, design of on-line scheduling). While, as in computer programming, hard-coding may allow for some benefit, it severely hurts reuse and understanding. For example, a LMI-based approach to on-line scheduling for linear systems may be impossible to adapt to nonlinear systems as LMI methods are typically restricted to linear systems. In contrast, abstractions provide for a clean interface which facilitates reuse and understanding, though at the cost of some pessimism.

\section{Application to General Weak Execution} \label{sec:general-weak}
To show the applicability of the convergence rate abstractions, we return to the particular case of weakly-hard execution discussed in \cref{sec:abstraction:problem}.

\begin{definition}[Linear Convergence Rate Abstraction for Weak Execution]
	\label{def:linear-weak-abstraction}
	In both this and the next section, we consider an abstraction of the form
	\begin{align}
	\bar v_{k+1} &= \rho_{\sigma_k} \bar v_k + \beta |w_k|, & \bar v_0 &= \alpha |x_0|
	\label{eq:abstraction-vbar-dynamics}
	\end{align}
	that guarantees
	\begin{equation}
	\forall k \geq 0: \quad |x_k| \le \bar v_k. \label{eq:abstraction-vbar-guarantee}
	\end{equation}
\end{definition}

The following subsections show how to determine appropriate parameters $\rho_{i}, \alpha, \beta > 0$ either from robust exponential stability (\cref{sec:weak:simple-abstraction}) or from Lyapunov functions of the nominal case (\cref{sec:weak:lyapunov-abstraction}). Note that an implementation does not need to compute these parameters or the abstraction state $v_k$ exactly. Instead, any upper approximation is also possible because it preserves the abstraction guarantee \eqref{eq:abstraction-vbar-guarantee}.

\subsection{Simple Robustness-Based Abstraction}\label{sec:weak:simple-abstraction}
To derive an abstraction based on robustness of the nominal case, we come back to the specific case discussed in \cref{sec:abstraction:problem} and its abstraction to $v$ (not $\bar v$) by \cref{eq:abstraction-dynamics,eq:abstraction-guarantee}. The results are equally valid for the nonlinear case or any modified setting, as long as \cref{eq:abstraction-dynamics,eq:abstraction-guarantee} are satisfied.

\begin{theorem}[Robustness-Based Abstraction]
	Let $\gamma_{i} \geq 0$, $i \in \Sigma$, be
	\begin{equation}
	\gamma_i=\|A_{i}-A_0\|_2 \label{eq:abstraction-gamma-for-linear-mk}
	\end{equation}
	in the linear case or, in general,
	a bound such that
	\begin{equation}
	\forall x_k, \sigma_k: \quad |A_{\sigma_k}(x_k) - A_0(x_k)| \le \gamma_{\sigma_k} |x_k|, \qquad \gamma_0=0.
	\label{eq:abstraction-gamma-def}
	\end{equation}
	Let $\rho,\alpha,\beta$ be parameters of the $v$-abstraction  \eqref{eq:abstraction-guarantee} -- \eqref{eq:abstraction-dynamics}. Then the same $\alpha,\beta$ and
	\begin{align}
		\rho_{\sigma_k} := \rho + \beta \gamma_{\sigma_k},
	\end{align}
	are valid parameters of the $\bar v$-abstraction \eqref{eq:abstraction-vbar-dynamics} -- \eqref{eq:abstraction-vbar-guarantee}.
\end{theorem}

\begin{proof}[Derivation and Proof]
	\item 
	\paragraph{Equivalence of disturbance and weak execution} The abstraction \eqref{eq:abstraction-dynamics} describes the nominal case with disturbance. In the following, it will be used for the case of disturbance plus weak execution, i.\,e., sometimes skipping the controller. For this, the deviation from $\sigma_k\equiv 0$ is interpreted as additional disturbance, leading to the new disturbance $\tilde w_k$:

	\begin{equation}
	x_{k+1}=\underbrace{A_{\sigma_k} (x_k) + w_k}_{\text{before: $\sigma_k$, $w_k$}} = \underbrace{\vphantom{A_{\sigma_k}} A_{0} (x_k)}_{\text{now: $\tilde \sigma_k=0$,}} + \underbrace{A_{\sigma_k}(x_k) - A_0(x_k) + w_k}_{\mathllap{\text{disturbance }}\tilde w_k}. \label{eq:weakly-hard-equivalence-to-disturbance}
	\end{equation}
	Therefore, we may equivalently replace $\sigma_k$ by $\tilde \sigma_k$ and $w_k$ by $\tilde w_k$ in the original system.

	The previous definition of $\gamma_i$ becomes
	\begin{equation}
	\forall x_k, \sigma_k: \quad |\underbrace{A_{\sigma_k}(x_k) - A_0(x_k)}_{\tilde w_k - w_k}| \le \gamma_{\sigma_k} |x_k|
	\stackrel{\text{\eqref{eq:abstraction-guarantee}}}{\le} \gamma_{\sigma_k} v_k, \qquad \gamma_0=0. \label{eq:abstraction-gamma-def-proof}
	\end{equation}

	\paragraph{Derivation of the abstraction}
	Then, the dynamics  of $v_k$ can be overapproximated by $\bar v_k$, forming a second layer of abstraction:
	Initialize $\bar v_k$ by
	\begin{align}
	\bar v_0 = & v_0 \stackrel{\text{\eqref{eq:abstraction-dynamics}}}{=} \alpha |x_0| \label{eq:abstraction-vbar-init}
	\intertext{so that the induction assumption (IA) $\bar v_k \geq v_k$ is satisfied for $k=0$. Then, the dynamics}
	\bar v_{k+1} &:= \underbrace{(\rho + \beta \gamma_{\sigma_k})}_{\rho_{\sigma_k}} \bar v_k + \beta |w_k|
	\\ &= \rho \underbrace{\bar v_k}_{\stackrel{\text{if (IA)}}{\geq} v_k } + \beta \underbrace{\gamma_{\sigma_k} \bar v_k}_{\stackrel{\text{if (IA)}}{\geq} \gamma_{\sigma_k} v_k \stackrel{\text{\eqref{eq:abstraction-gamma-def-proof}}}{\geq} |\tilde w_k - w_k|} + \beta |w_k|
	\\ & \stackrel{\mathclap{\text{if (IA)}}}{\geq}~~ \rho v_k + \beta |\tilde w_k - w_k|  + \beta |w_k|\\
	&\ge  \rho v_k + \beta |\tilde w_k|\\
	&\stackrel{\mathclap{\text{\eqref{eq:abstraction-dynamics} for $\tilde\sigma,\tilde w$ per \eqref{eq:weakly-hard-equivalence-to-disturbance}}}}{=} \hspace{4em} v_{k+1}
	\intertext{guarantee by induction that}
	\forall k \geq 0: \quad \bar v_k &\geq v_{k}
	\intertext{and therefore}
	\forall k \geq 0: \quad |x_k| & \stackrel{\text{\eqref{eq:abstraction-guarantee}}}{\le} \bar v_{k}.
	\end{align}
\end{proof}

\subsection{Improved Abstraction for Weak Execution Based on Lyapunov Functions}
\label{sec:weak:lyapunov-abstraction}
The previous abstraction treats any deviation from the nominal case as disturbance, which may cause pessimism. As an extreme example, consider the case of $A_0 \neq 0, A_1 = 0$, i.e., the state jumps to zero immediately for the non-nominal execution $\sigma_k=1$. Then, as detailed in the previous derivation, this behavior is abstracted as the worst behavior possible from any disturbance with magnitude $|w_k|\le|A_{1}(x_k)-A_0(x_k)|$, i.e., as possibly increasing the state instead of actually zeroing it.

An improved abstraction which represents the non-nominal case with better accuracy is possible by considering a quadratically bounded Lyapunov function $V(x)$ for the nominal case and abstracting the state space by level sets $V(x) \le \mathrm{const}$. The theory behind this abstraction is closely related to the proof sketch of \cref{thm:exp-decay} for the nonlinear case.

\begin{theorem}

	Consider the system
	\begin{align}
	x_{k+1} &= A_{\sigma_k}(x_k) + w_k, \quad x_0 \in \R^n. \label{eq:improved-mk-abstr:sys}
	\end{align}
	Let $V(x)$ be any function which fulfills the following properties, e.g., $V(x) = x^\transp P x$ with positive definite $P \in \R^{n \times n}$:
	\begin{align}
	\forall x \in \R^n: \quad V(x)&
	\begin{cases}
	>0, & x \neq 0,\\
	=0, & x = 0
	\end{cases} & \text{ (positive definite),}\\
	\exists \gamma \ge 1: \quad \forall a,b \in \R^n: \quad \sqrt{V(a+b)} \le& \sqrt{V(a)}+\gamma \sqrt{V(b)}  & \text{ (weakly subadditive),} \label{eq:improved-mk-abstr:subadd}\\
	\exists c_1,c_2 >0: \quad \quad \forall x \in \R^n: \quad c_1 \sqrt{V(x)} \le& |x| \le c_2 \sqrt{V(x)} & \text{(equivalent to norm).} \label{eq:improved-mk-abstr:equiv}
	\end{align}
	Then,
	\begin{align}
		\bar v_{k+1} &= \rho_{\sigma_k} \bar v_k + \beta |w_k|, & \bar v_0 &= \alpha |x_0| \label{eq:improved-mk-abstr:abstraction-dynamics}
	\end{align}
		is a convergence rate abstraction for the system \eqref{eq:improved-mk-abstr:sys}, i.e., it guarantees $|x_k| \le \bar v_k$, for
	\begin{align}
		 \rho_{\sigma_k} &\geq \|A_{\sigma_k}\|_V := \sup_{x \neq 0} \sqrt{\frac{V(A_{\sigma_k}(x))}{V(x)}},& \alpha&\geq \frac{c_2}{c_1},& \beta \ge \gamma \frac{c_2}{c_1}. \label{eq:improved-mk-abstr:alpha-rho-norm-def}
	\end{align}
\end{theorem}

The above requirements are fulfilled by a wide class of Lyapunov functions including quadratic forms $V(x)=x^\transp P x$, the square $V(x)=\|x\|_v^2$ of an arbitrary vector norm $\|\cdot\|_v$, and piecewise-defined variants thereof, such as piecewise-ellipsoidal or piecewise-polytopic functions. The connection between convergence rate abstractions and Lyapunov functions is discussed later.

\begin{proof}[Proof]
	\begin{align}
	|x_{k+1}| \stackrel{\text{\eqref{eq:improved-mk-abstr:equiv}}}{\le} &c_2 \sqrt{V(x_{k+1})}\\ \stackrel{\text{\eqref{eq:improved-mk-abstr:sys}}}{=} &c_2 \sqrt{V(A_{\sigma_k} (x_k) + w_k})\\ \stackrel{\text{\eqref{eq:improved-mk-abstr:subadd}}}{\le} &c_2 \sqrt{V(A_{\sigma_k} (x_k))} + c_2 \gamma \sqrt{V(w_k)} \\
	\stackrel{\text{\eqref{eq:improved-mk-abstr:alpha-rho-norm-def}}}{\le} &c_2 \|A_{\sigma_k}\|_V \sqrt{V(x_k)} +c_2  \gamma \sqrt{V(w_k)} \\
	\stackrel{\text{\eqref{eq:improved-mk-abstr:equiv}}}{\le} &c_2 \|A_{\sigma_k}\|_V \sqrt{V(x_k)} + \frac{c_2}{c_1} \gamma |w_k| \label{eq:improved-mk-abstr:x-bound}
	\end{align}

	The abstraction's state $\bar v_k$ tracks an upper bound of $c_2 \sqrt{V(x_k)}$ arising from this equation, which is shown by induction in the following:

	\paragraph{Induction assumption IA($k$):}
	\begin{equation}
		\text{IA($k$)} \quad :\Leftrightarrow \quad c_2 \sqrt{V(x_k)} \le \bar v_k
	\end{equation}

	\paragraph{Start of induction:}
	\begin{equation}
	\bar v_0 \stackrel{\text{\eqref{eq:improved-mk-abstr:abstraction-dynamics}}}{=} \alpha |x_0| \stackrel{\text{\eqref{eq:improved-mk-abstr:equiv}}}{\ge} \alpha c_1 \sqrt{V(x_0)} \stackrel{\text{\eqref{eq:improved-mk-abstr:alpha-rho-norm-def}}}{\ge} c_2 \sqrt{V(x_0)} \Rightarrow  \text{ IA($0$)}
	\end{equation}

	\paragraph{Induction step:}
	\begin{align}
		\text{IA($k$)} \Rightarrow |x_{k+1}| \stackrel{\text{\eqref{eq:improved-mk-abstr:x-bound}}}{\le} &
		c_2 \|A_{\sigma_k}\|_V \sqrt{V(x_k)} + \frac{c_2}{c_1} \gamma |w_k|\\
		\stackrel{\text{IA($k$)}}{\le}& \|A_{\sigma_k}\|_V \bar v_k  + \frac{c_2}{c_1} \gamma |w_k|\\
		\stackrel{\text{\eqref{eq:improved-mk-abstr:alpha-rho-norm-def}}}{\le} & \rho_{\sigma_k} \bar v_k  + \beta |w_k| \stackrel{\text{\eqref{eq:improved-mk-abstr:abstraction-dynamics}}}{=} v_{k+1}
		\Rightarrow \text{ IA($k+1$)}
	\end{align}

	\paragraph{Conclusion:}
	By induction, IA($k$) holds for all $k \geq 0$. This proves the desired abstraction guarantee since
	\begin{equation}
	|x_k| \stackrel{\text{\eqref{eq:improved-mk-abstr:equiv}}}{\le} c_2 \sqrt{V(x_0)} \stackrel{\text{IA($k$)}}{\le} \bar v_k \quad \forall k \geq 0.
	\end{equation}

\end{proof}

The best bound for $\alpha$ is
\begin{equation}
	\alpha^*=\frac{c_{2,\min}}{c_{1,\max}} = \frac{\sup\limits_{x \neq 0}\frac{|x|}{\sqrt{V(x)}}}{\inf \limits_{x \neq 0}\frac{|x|}{\sqrt{V(x)}}}\stackrel{(*)}{=}\frac{\sup\limits_{x \neq 0}\frac{\sqrt{V(x)}}{|x|}}{\inf \limits_{x \neq 0}\frac{\sqrt{V(x)}}{|x|}},
\end{equation}
where $(*)$ is due to $\sup{|x|} = 1/\inf{|\frac{1}{x}}|$ for $x \neq 0$.
\paragraph{Interpretation for Quadratic Lyapunov Functions}
For $V(x)=x^\transp P x$ with $P$ positive definite, the level set $V(x)=\mathrm{const}$ is an ellipsoid, so the abstraction can be described as \emph{ellipsoid abstract domain}; see \cite{Roux2012} for a detailed discussion. Then, $\alpha^*$ is the excentricity of the ellipsoid, i.e., the ratio of major and minor axis; $\sqrt{x^\transp P x}$ is a vector norm, so $\gamma=1$ and, in the linear case, $\|A_\sigma\|_V$ is its induced \emph{ellipsoidal matrix norm} (\cite[Section~2.3.7]{Jungers2009}, \cite[Section~8]{Gaukler2019extended}).

\paragraph{Connection to Lyapunov Functions}
To derive an abstraction that shows stability of the nominal case using the previous theorem, it is required that $
\|A_{0}\|_V < 1$. This equivalently means that $V(x)$ is a Lyapunov function for the \emph{nominal case} because
\begin{align}
V(A_0(x)) &\le \underbrace{\|A_0\|_V^2}_{<1} V(x)\\
\Rightarrow ... \Rightarrow V(x_k) &\le \|A_0\|_V^{2k} V(x_0) \stackrel{k \to \infty}{\to} 0 \quad\text{ for }\sigma \equiv 0, w \equiv 0\\
\stackrel{\text{\eqref{eq:improved-mk-abstr:equiv}}}{\Rightarrow} |x_k| &\stackrel{k \to \infty}{\to} 0 \text{ for }\sigma \equiv 0, w \equiv 0.
\end{align}
However, the same does not hold for weakly-hard execution: If $\|A_1\|_V>1$, then $V(x)$ may increase on $\sigma_k=1$, so it is no longer a Lyapunov function for the \emph{weakly-hard} system. Nevertheless, as will be discussed in the next sections, stability can still be shown if the abstraction (or equivalently, $V(x)$) is decreasing on average, \ie, in the long term. Then, $V(x)$ is a \emph{Lyapunov-like function} for the weakly-hard system, which decreases in the long term but may temporarily increase, similar to the definition in \cite[Theorem 4.2]{Ye1996}.

\subsection{Stability Verification}

\begin{theorem}[Abstracted Stability Criterion for Weak Execution]
	\label{thm:abstraction-exp-stability-kappa}
With the factor
\begin{equation}
 \kappa_{a,b} \coloneq \frac{\bar v_b}{\bar v_a}  = \prod_{i=a}^{b-1} \rho_{\sigma_i}, \quad a \le b, \label{eq:abstraction-kappa-def}
\end{equation}
a sufficient criterion for exponential stability can be constructed:

If the execution sequence $\sigma_k$ satisfies
\begin{equation}
	\forall k \geq 0 \quad \kappa_{0,k} \leq \tilde \alpha \tilde \rho^k \label{eq:abstraction-exp-stability-kappa-condition}
\end{equation}
with $\tilde \rho  \in \lbrack 0;1)$ and $\tilde \alpha \geq 1$,
then the original system \eqref{eq:abstraction-orig-system} is exponentially stable for $w_k \equiv 0$.
\end{theorem}
\begin{proof} Assume $\tilde \rho  \in \lbrack 0;1)$ and $w_k \equiv 0$.
\begin{align}
\text{\eqref{eq:abstraction-exp-stability-kappa-condition}} &\Rightarrow \left(
\forall k \geq 0,x_0 \qquad |x_k|   \stackrel{\text{\eqref{eq:abstraction-vbar-guarantee}}}{\le} \bar v_k
	 \stackrel{\text{\eqref{eq:abstraction-kappa-def}}}{=} \kappa_{0,k} \bar v_0
	 \stackrel{\text{\eqref{eq:abstraction-exp-stability-kappa-condition}, \eqref{eq:abstraction-vbar-init}}}{\le} \tilde \alpha \tilde \rho^k \alpha |x_0|
\right)  \nonumber \\
	  &\Rightarrow \text{\eqref{eq:abstraction-orig-system} is exponentially stable.}
\label{eq:abstraction-exp-stability-kappa}
\end{align}
\end{proof}

\section{Application to $(m,K)$-weak execution} \label{sec:mk-weak}
This section will specialize the results of the previous section to the case of $(m,K)$-weak execution, which is considered in most work on weakly-hard control systems suffering from packet loss or deadline misses. Throughout this section, we again consider the abstraction given by \cref{def:linear-weak-abstraction}.

\begin{definition}[$(m,K)$-weak execution \cite{Hamdaoui1995}]
	\label{def:mk}
	In this case, there are only two execution modes ($\Sigma=\{0,1\}$). In any $K$ consecutive control periods, at least $m$ controller executions are executed normally ($\sigma_i=0$), while the remaining up to $\bar m \coloneq K-m$ controller executions are skipped ($\sigma_i=1$):
	\begin{equation}
	\forall k\geq 0: \quad \sum_{i=k}^{k+K-1} \sigma_i \le \bar m \coloneq K-m. \label{eq:abstraction-mk-def}
	\end{equation}
\end{definition}
For the later derivations, an upper bound on the number of skips in $k$ periods is constructed by partitioning the sequence $i=0,1,\dots,k-1$ into $\lfloor k/K \rfloor$ chunks of length $K$ (first: $0 \le i \le K-1$, second: $K \le i \le 2K-1$, and so on) and a remainder ($k - (k \bmod K) \le i \le k-1$)  of length $k \bmod K$.
\begin{align}
\nonumber
\Rightarrow \qquad 	\forall k\geq 0:\quad \sum_{i=0}^{k-1} \sigma_i =& \sum_{j=1}^{\lfloor k/K \rfloor} \underbrace{ \sum_{i=(j-1)K}^{jK - 1} \sigma_i}_{\le \bar m \text{ due to \eqref{eq:abstraction-mk-def}}} + \sum_{i=k - (k \bmod K)}^{k-1} \underbrace{\sigma_i}_{\le 1} \\
\le& \underbrace{\bar m \lfloor k/K \rfloor}_{\text{$\bar m$ per integer multiple of $K$}}  + \underbrace{\min(\bar m,  k \bmod K)}_{\text{remaining $<K$ periods: at most $\bar m$}}. \label{eq:abstraction-mk-def-implication}
\end{align}

This bound is tight as it is reached for $\sigma_k=\begin{cases}
1,& k \bmod K \le \bar m\\
0,& \text{else.}
\end{cases}$

\begin{proof}[Derivation of an Exponential Stability Criterion for $(m,K)$-Weak Execution] Assume $\sigma$ obeys $(m,K)$-weak execution and $\rho_1 \geq \rho_0$. Evaluate $\kappa_{0,k}$ from \cref{thm:abstraction-exp-stability-kappa}:
\begin{equation}
	 \kappa_{0,k} = \prod_{i=0}^{k-1} \rho_{\sigma_i} \stackrel{\text{\eqref{eq:abstraction-gamma-def}}}{=} \prod_{i=0}^{k-1}
	 \begin{cases}
	 \rho_0, & \sigma_i=0\\
	 \rho_1,& \sigma_i=1\\
	 \end{cases}.
\end{equation}
Using $\sigma_i \in \left\lbrace0,1\right\rbrace$ and the $(m,K)$-constraint~\eqref{eq:abstraction-mk-def}, we can derive an upper bound $\bar \kappa_{0,k}$:
\begin{align}
	\kappa_{0,k} &= \rho_0^{\overbrace{\scriptstyle \sum_{i=0}^{k-1} (1-\sigma_i)}^{\mathclap{k-\sum_{i=0}^{k-1} \sigma_i}}} ~\rho_1^{\sum_{i=0}^{k-1} \sigma_i}
	= \rho_0^k \big(\underbrace{\rho_0^{-1} \rho_1}_{\mathclap{\text{assumed }\geq 1}} \big)^{\overbrace{\scriptstyle \sum_{i=0}^{k-1} \sigma_i}^{\mathclap{\text{upper bounded by \eqref{eq:abstraction-mk-def-implication}}}}} \\
	& \stackrel{\text{\eqref{eq:abstraction-mk-def-implication}}}{\leq}
	 \rho_0^k \left(\rho_0^{-1} \rho_1\right)^{\bar m \lfloor k/K \rfloor  + \min(\bar m,  k \bmod K)}
	 \eqqcolon \bar \kappa_{0,k}. \label{eq:abstraction-kappabar}
\end{align}
For $k \to \infty$, the maximum average ratio of skips ($\sigma_k=1$) is $\bar m/K$:
\begin{align}
	\lim_{k \to \infty} \frac{1}{k} \sum_{i=0}^{k-1} (\sigma_i) \stackrel{\text{\eqref{eq:abstraction-mk-def-implication}}}{\le}& \lim_{k\to\infty} \frac{\bar m \lfloor k/K \rfloor  + \min(\bar m,  k \bmod K)}{k} \nonumber \\ 
	=& \lim_{k\to\infty} \Bigg( \bar m \underbrace{\frac{\lfloor k/K \rfloor}{k}}_{\to \frac{1}{K}}  + \underbrace{\frac{\min(\bar m,  k \bmod K)}{k}}_{\to 0} \Bigg)= \frac{\bar m}{K}. \label{eq:abstraction-avg-droprate}
\end{align}

The following derivation uses the above statements to determine the minimum $\tilde\rho$ and corresponding $\tilde \alpha$ that fulfill the ansatz
\begin{equation}
	\bar \kappa_{0,k} \stackrel{!}{\leq} \tilde \alpha \tilde \rho^k \quad \forall k \ge 0,
\end{equation}
which shall later be used to show the abstracted stability criterion \eqref{eq:abstraction-exp-stability-kappa}.
Assuming $\tilde \rho>0$, the ansatz is equivalent to
\begin{equation}
	\bar \kappa_{0,k}^{1/k} \stackrel{!}{\leq} \underbrace{\tilde \alpha^{1/k}}_{\mathrlap{\hspace{-.5em}\to 1 \text{ for }k \to \infty}} \tilde \rho,
\end{equation}
which motivates that a candidate for the minimum $\tilde \rho$ is
\begin{equation}
	\tilde \rho \coloneq \lim_{k\to\infty} \bar \kappa_{0,k}^{1/k} \stackrel{\text{\eqref{eq:abstraction-kappabar}}}{=} \lim_{k\to\infty}  \rho_0 \left(\rho_0^{-1} \rho_1\right)^{\frac{\bar m \lfloor k/K \rfloor  + \min(\bar m,  k \bmod K)}{k}} \stackrel{\text{\eqref{eq:abstraction-avg-droprate}}}{=} \rho_0^{1-\frac{\bar m}{K}} \rho_1^{\frac{\bar m}{K}}
	\stackrel{\text{\eqref{eq:abstraction-mk-def}}}{=}
	\rho_0^{\frac{m}{K}} \rho_1^{\frac{K-m}{K}}
	 \label{eq:abstraction-derivation-rho-kappabar}.
\end{equation}
The result can be interpreted as a special \enquote{weighted average} of the stability exponent for the extreme cases: Never skipping ($m=K$) yields the nominal case $\tilde \rho=\rho_0$ and always skipping ($m=0$) results in $\tilde \rho = \rho_1$.

The validity of this candidate $\tilde \rho$ is then implicitly proven by determining the corresponding overshoot factor $\tilde \alpha$ and showing that $\tilde \alpha < \infty$: For $\tilde \alpha$, consider
\begin{equation}
\bar \kappa_{0,k} \stackrel{!}{\leq} \tilde \alpha  \tilde \rho^k \quad \forall k \geq 0\quad \Leftrightarrow \quad \tilde \alpha \geq \bar \kappa_{0,k} \tilde \rho^{-k} \quad \forall k \geq 0\quad \Leftrightarrow \quad \tilde \alpha \geq \max_{k \geq 0}  \bar \kappa_{0,k} \tilde \rho^{-k} \label{eq:abstraction-derivation-alpha-kappabar-equivalence}
\end{equation}
and choose $\alpha$ as a finite upper bound for the right hand side:
\begin{align}
	\max_{k \geq 0}  \tilde \rho^{-k} \bar \kappa_{0,k}
	\stackrel{\text{\eqref{eq:abstraction-kappabar}, \eqref{eq:abstraction-derivation-rho-kappabar}}}{=}& \max_{k \geq 0} \rho_0^{-\frac{K-\bar m}{K}k} \rho_1^{-\frac{\bar m}{K}k} \, \rho_0^k \big(\underbrace{\rho_0^{-1} \rho_1}_{\mathclap{\text{assumed $\geq 1$}}}\big)^{\bar m \overbrace{\scriptstyle \lfloor k/K \rfloor}^{\leq k/K}  + \overbrace{\scriptstyle \min(\bar m,  k \bmod K)}^{\leq \bar m}}  \nonumber \\
	\leq& \max_{k \geq 0} \rho_0^{-\frac{K-\bar m}{K}k} \rho_1^{-\frac{\bar m}{K}k} \, \rho_0^k \left(\rho_0^{-1} \rho_1\right)^{\bar m k/K  + \bar m} \nonumber \\
	=& \max_{k \geq 0} \rho_0^{k\overbrace{\scriptstyle \left(-\frac{K-\bar m}{K} + 1 - \frac{\bar m}{K} \right)}^{0} - \bar m} \, \rho_1^{k\overbrace{\scriptstyle  \left( -\frac{\bar m}{K} + \frac{\bar m}{K}\right)}^{0} + \bar m} \nonumber \\
	=& \rho_0^{-\bar m}\rho_1^{\bar m} = \left(\frac{\rho_1}{\rho_0}\right)^{\bar m} \stackrel{\text{\eqref{eq:abstraction-mk-def}}}{=} \left(\frac{\rho_1}{\rho_0}\right)^{K - m} \eqcolon \tilde \alpha. \label{eq:abstraction-derivation-alpha-kappabar}
\end{align}

Note that $\tilde \alpha = \left(\tilde \rho/\rho_0\right)^K$.

With $\tilde \alpha$ and $\tilde \rho$ as per \cref{eq:abstraction-derivation-alpha-kappabar,eq:abstraction-derivation-rho-kappabar}, we can prove the ansatz
\begin{equation}
	\forall k\geq 0: \quad \kappa_{0,k} \stackrel{\text{\eqref{eq:abstraction-kappabar}}}{\leq} \bar \kappa_{0,k} \stackrel{\text{\eqref{eq:abstraction-derivation-alpha-kappabar-equivalence}, \eqref{eq:abstraction-derivation-alpha-kappabar}}}{\leq}\tilde \alpha \tilde \rho^k, \label{eq:abstraction-exp-stability-kappabar-bound}
\end{equation}
which leads to the stability criterion:
\end{proof}

\begin{theorem}[Exponential Stability Criterion for $(m,K)$-weak Execution]
	\label{thm:exp-stability-mk}
The above derivation shows that
\begin{equation}
	\underbrace{\rho_0^{\frac{m}{K}} \rho_1^{\frac{K-m}{K}}}_{\displaystyle \tilde \rho} < 1 \stackrel{\text{\eqref{eq:abstraction-exp-stability-kappa}, \eqref{eq:abstraction-exp-stability-kappabar-bound}}}{\Rightarrow}
	\begin{cases}
	\text{For $w_k \equiv 0$, $\rho_1 \geq \rho_0$ and $(m,K)$-weak execution (\cref{def:mk}),} \\
	\text{the original system \eqref{eq:abstraction-orig-system} is exponentially stable:}\\
	\forall k \geq 0,x_0 \quad |x_k| \le \alpha \tilde \alpha \tilde \rho^k |x_0|\\
	\text{with $\tilde \rho$, $\tilde \alpha$ per \labelcref{eq:abstraction-derivation-alpha-kappabar,eq:abstraction-derivation-rho-kappabar} and $\rho_0,\rho_1,\alpha$ per \cref{def:linear-weak-abstraction}.}
	\end{cases}
	 \label{eq:abstraction-exp-stability-kappabar}
\end{equation}
\end{theorem}

This criterion exemplifies the benefit of convergence rate abstractions: Stability can be shown without the intricate computation of an explicit stability certificate for the \emph{weakly-hard} system, such as a Lyapunov function or reachable set. Instead, the criterion only requires the execution parameters $(m,K)$ and a simple abstraction summarizing the stability and robustness of the \emph{nominal} system. The detailed system dynamics are not required, since they are abstracted by three scalar parameters $\rho_0, \alpha$ and $\rho_1$, which model exponential decay, initial overshoot and the sensitivity to skipping the controller execution. Determining these is possible merely from the exponential stability of the nominal case (\cref{sec:weak:simple-abstraction}) or, optionally, via Lyapunov functions for the \emph{nominal} case (\cref{sec:weak:lyapunov-abstraction}). Both methods are considerably easier than directly analyzing the weakly-hard case.

To determine the permissible skip ratio $(K-m)/K$ for a desired convergence rate $\tilde \rho$, \eqref{eq:abstraction-derivation-rho-kappabar} may be rewritten as
\begin{align}
\tilde \rho &= \rho_0^{1-\frac{\bar m}{K}} \rho_1^{\frac{\bar m}{K}} = \rho_0 \left(\rho_1 \rho_0^{-1}\right)^{\frac{\bar m}{K}}\\
\Leftrightarrow ~ \frac{K-m}{K} \stackrel{\text{\eqref{eq:abstraction-mk-def}}}{=} \frac{\bar m}{K} &= \frac{\log \left(\frac{\tilde \rho}{\rho_0}\right)}{\log \left(\frac{\rho_1}{\rho_0} \right)}.
\end{align}

It is interesting to note that the stability result provided by this criterion does not depend on the actual value of $K$, but on the skip ratio $(K-m)/K$ or, equivalently, the execution ratio $m/K$: Increasing $K$ and $m$ proportionally only increases the overshoot $\tilde \alpha$, but not the growth rate $\tilde \rho$, as the latter only depends on the ratio. Therefore, the criterion shows stability for $(m, K)$-weak execution if and only if it shows stability for $(cm, cK)$-weak execution, where $c \in \N$ is an arbitrary integer.

The generality of the proposed convergence rate abstractions can be seen by the fact that it contains existing results as a special case. For example, in \cite[Theorem 2]{Horssen2016} a Linear Matrix Inequality equivalent to the quadratic Lyapunov-like function $V(x)=x^\transp P x$ is used to show stability if, in our terms\footnote{Note that here, $\sigma=0$ means nominal execution and $\sigma=1$ means skipping. In \cite{Horssen2016}, it is the opposite.}, $\rho_1^{K-m} \rho_0^{m} < 1$, which is equivalent to $\tilde \rho < 1$. In that context, the criterion is shown to be conservative, which leads to the following question:

\paragraph{On the General Existence of Stable Abstractions}
It is an interesting open question under which conditions there is a converse variant of  \cref{thm:exp-stability-mk}: If a system is stable under weak execution, under which conditions does a stable abstraction exist? How complicated does this abstraction have to be (\eg, nonlinear or multi-dimensional)? Are there simplifications for typical practical cases? In general, these are open questions. As a first step, we will now show that simple abstractions, \ie, matching \cref{def:linear-weak-abstraction}, do not necessarily exist:

\paragraph{Naive Expectation: Converse Stability Criterion for $(m,K)$-weak Execution}
Assume a linear system without disturbance ($w_k \equiv 0$) that is exponentially stable under $(m,K)$-weak execution. One could expect that for any such system there is an abstraction
\begin{equation}
	\bar v_{k+1} = \rho_{\sigma_k} \bar v_k,\quad \bar v_0 \geq \bar \alpha |x_0| \label{eq:abstraction-exp-mk-converse:vk}
\end{equation} with constants $\rho_0, \rho_1, \bar \alpha$ such that
\begin{enumerate}
	\item The abstraction is valid, \ie, $|x_k| \le \bar v_k$ for all $x_0$ and all $(m,K)$-executions $(\sigma_0, \sigma_1,\dots)$.
	\item The abstraction proves exponential stability by \eqref{eq:abstraction-exp-stability-kappabar}, \ie, $\tilde \rho = \rho_0^{\frac{m}{K}} \rho_1^{\frac{K-m}{K}} < 1$.
\end{enumerate}

This would be helpful as it would imply that \cref{thm:exp-stability-mk} is sufficient and necessary, so that stability under $(m,K)$-execution is \emph{always} equivalent to stability under $(cm, cK)$-execution for $c \in \N$. However, the following academic counterexample will show that this is not generally true, which also matches the conservatism stated in \cref{remark:abstraction-weakly-hard:conservative} for a rather general and in \cite[Theorem 2]{Horssen2016} for a specific case.

Therefore, at least this simple variant of an abstraction is conservative, which raises some further open questions: How significant is this conservatism in practice? Can it be reduced by simple extensions, \eg, by having $\rho_{\sigma_k}$ depend not only on $\sigma_k$ but also on $\sigma_{k-1}$?

\begin{proof}[Counterexample]
	The system
	\begin{align}
		A_0&=\mat{a & 0 & a\\ 0 & 0 & 0 \\ 0 & 0 & 0}\\
		A_1&=\mat{a & 0 & 0 \\ c & 0 & 0 \\ 0 & 1 & 0}\\
		a&=\frac{1}{2},\quad c=1000
	\end{align}
	is stable under $(m,K)$-execution for $m=1,K=2$, but not for $m=2,K=4$, as will be shown below. Therefore, the above naive expectation is false because it would imply that $(1,2)$-stability is equivalent to $(2,4)$-stability.

	For the following steps, symbolic computations and a numerical experiment can be found in \verb|counterexample_mk_abstraction.m|, which is available in the anciliary files of this arXiv.org publication. To formally denote $(m,K)$-sequences and the resulting transition matrices, we introduce the notation $\{1,2,3\}$, which means that any of the given numbers may be inserted at any place; \eg, $A^{\{0,1\}}B^{\{0,1\}}$ may be $I$, $A$, $B$ or $AB$. Similarly, $A^{\{1,2,\dots\}}$ may be any $A^i$ with $i \geq 1$.
Which numbers will be inserted depends on the activation sequence $(\sigma_{0},\sigma_{1},\dots)$.

	Consider $m=1,K=2$, which results in the activation sequence
	\begin{align}
	(\sigma_0, \sigma_1, \dots,\sigma_{k-1}) &= (\underbrace{0,0,\dots,0,}_{\{0,1,\dots\} \text{ times}}\overbrace{1,\underbrace{0,0,\dots,0,}_{\{1,2,\dots\} \text{ times}}1,\underbrace{0,0,\dots,0,}_{\{1,2,\dots\} \text{ times}} \dots,}^{\{0,1,\dots\} \text{ repetitions of \enquote{1000...0}}} \underbrace{1}_{\text{\{0,1\} times}}),\\
	\intertext{and the state evolution $x_k=\Phi_{k}x_0$ with}
	\Phi_k &:= A_{\sigma_{k-1}} \cdots A_{\sigma_{1}} A_{\sigma_{0}} \\
	&= \overbrace{A_1^{\{0,1\}}\underbrace{A_0^{\{1,2,\dots\}} A_1 A_0^{\{1,2,\dots\}}A_1\cdots}_{0,1,\dots\text{ terms }A_0^{\{1,2,\dots\}} A_1}   A_0^{\{0,1,2,\dots\}}}^{k\text{ matrices}} ~.
	\intertext{In the above term, $A_0$ and $A_0A_1$ are \enquote{stable} factors:}
		\|A_0\|_2 &= \sqrt{2} a  \approx 0.707 < 0.9\\
		\|A_1\|_2 &= \sqrt{a^2+c^2} \approx 1000 < 1001\\
		\|A_0A_1\|_2 &= \sqrt{a^4+a^2} \approx 0.559 < 0.81 = 0.9^2\\
		\intertext{Using these norms and submultiplicativity, $\Phi_k$ can be bounded:}
		\Rightarrow  \|\Phi_k\|_2 &= \|\overbrace{A_1^{\{0,1\}}\underbrace{A_0^{\{1,2,\dots\}} A_1 A_0^{\{1,2,\dots\}}A_1\cdots}_{0,1,\dots\text{ terms }A_0^{\{1,2,\dots\}} A_1}   A_0^{\{0,1,2,\dots\}}}^{k\text{ matrices}}\|_2 \\
		&\le \|A_1\|_2^{\{0,1\}}  \|A_0^{\{0,1,\dots\}}\|_2 \|A_0A_1\|_2 \|A_0^{\{0,1,\dots\}}\|_2 \|A_0A_1\|_2 \dots \|A_0\|_2^{\{0,1,\dots\}} \\
		&\le \|A_1\|_2^1 \underbrace{0.9^{\{0,1,\dots\}} 0.9^2 0.9^{\{0,1,\dots\}} 0.9^2 \dots 0.9^{\{0,1,\dots\}}}_{k\text{ factors}}\\
		&\le 1001 \cdot 0.9^k.\\
	\end{align}
	The above implies $|x_k| < 0.9^k \cdot 1001 |x_0|$, so the system is exponentially stable for $m=1,K=2$.

	However, $(m=2,K=4)$-execution is unstable for the sequence
	\begin{align}
		(\sigma_0, \sigma_1, \dots)&=(0,0,1,1,0,0,1,1,\dots).
	\end{align}
	Consider $k=4i$:
	\begin{align}
		\Phi_{4i} &= (A_1^2A_0^2)^{i},\quad  \\
		A_1^2 A_0^2 &= \mat{a^4 & 0 & a^4\\ a^3\,c & 0 & a^3\,c\\ a^2\,c & 0 & a^2\,c} \mat{a^2 & 0 & a^2\\ 0 & 0 & 0\\ 0 & 0 & 0 } = \mat{a^4 & 0 & a^4\\ a^3\,c & 0 & a^3\,c\\ a^2\,c & 0 & \smashUnderbrace{a^2\,c}{\gg 1} }
	\end{align}
	Consider the eigenvalues $\lambda$ of this matrix to see that the system is unstable:
	\begin{align}
		\Rightarrow \lambda\{A_1^2A_0^2\} &= \{0, 0, \underbrace{a^4 + c\,a^2}_{>250} \}\\
	\Rightarrow \lambda\{\Phi_{4i}\} = \lambda\{(A_1^2A_0^2)^i\}& = \{0, 0, \underbrace{(a^4 + c\,a^2)^i}_{>250^i}\}
	\end{align}
	Choose $x_0$ as the eigenvector of $A_1^2A_0^2$ corresponding to the eigenvalue $\lambda >250$. Then,
	\begin{equation}
		|x_{4i}|=|(A_1^2A_0^2)^i x_0| = |\lambda^i x_0|>250^i |x_0|,
	\end{equation}
	so the system is unstable for $(2,4)$-execution.

	To check the above results numerically, the maximum averaged spectral radius
	\begin{equation}
		\hat \rho_L := \max_{\substack{\sigma_0,\dots,\sigma_{L-1}\\ \text{satisfying $(m,K)$}}} \max_i |\lambda_i\{\Phi_L\}|^{1/L}
	\end{equation}
	over all $(m,K)$-sequences $\sigma$ of length $L=24$ was computed for the specific values of $m$ and $K$. This averaged spectral radius approximates $\rho$ analogously to the definition of the Joint Spectral Radius \cite[Chapter~1.1]{Jungers2009}. The results support the statements given in the proof: For $(1,2)$-weak execution, $\hat \rho \approx 0.71 < 0.9$, and for $(2,4)$, the result is $\hat \rho \approx 3.9766 > 3.97635... = (250)^{1/4}$.
\end{proof}

\paragraph{Generalization to non-global exponential stability} For nonlinear systems, global exponential stability is a strong requirement which is often not fulfilled. If the original system (with perfect execution and zero disturbance) is only exponentially stable within some initial set $x_0 \in X_0$, it is required to check that $x_k$ does not escape the stability region during overshoots caused by skipping the controller. This means that exponential stability of the weakly-hard system will only be valid in a smaller initial set $x_0 \in \tilde X_0$. (For better understanding, assume that $X_0$ is chosen as an invariant set, \ie, in the nominal case $x_0 \in X_0 ~\Rightarrow~ \forall k: x_k \in X_0$ )

Because disturbing the system by skipping the controller corresponds to a reinitialization of the original dynamics, it must be ensured that then the state still is inside the initial set from which exponential convergence is guaranteed, which motivates that the safe initial set under weakly-hard execution will typically be smaller than the original safe set $X_0$.

The above holds for any form of weakly-hard execution. The following discussion gives a concrete example for $(m,K)$-weak execution. In this case,
an upper bound on the state can be determined from the stability result \eqref{eq:abstraction-exp-stability-kappabar}:
\begin{equation}
	|x_k| \stackrel{\text{\eqref{eq:abstraction-exp-stability-kappabar}}}{\le} \alpha \tilde \alpha \tilde \rho^k |x_0| \le  \alpha \tilde \alpha |x_0|.
	\label{eq:abstraction-nonglobal-stability-statebound}
\end{equation}

We require $X_0 \supseteq \mathcal B_{r_0}$.
To ensure that $x_k$ never leaves the set $X_0$ in which the stability assumption holds, the new set $\tilde X_0$ must be limited to
\begin{equation}
 \tilde X_0 = \mathcal B_{r_0 (\alpha \tilde{\alpha})^{-1}}.
\end{equation}

This is formalized by the following theorem:

\begin{theorem}[$(m,K)$-Weak Stability Criterion for Non-Global Exponential Stability]
	Consider a nonlinear system which is not globally exponentially stable, but at least within $|x| \le r_0$. In formal terms, assume that the abstraction of $x_{k+1}$ by $v_{k+1}$ with \cref{eq:abstraction-guarantee,eq:abstraction-dynamics} is valid if $|x_k| \le r_0$ or equivalently if $v_k \le r_0$.
	Then, the stability result \eqref{eq:abstraction-exp-stability-kappabar} holds for all initial states within $|x_0| \le r_0 \alpha^{-1} \tilde{\alpha}^{-1}$.
\end{theorem}
\begin{proof}[Proof Sketch]
	Assume that $\tilde \rho < 1$ in \eqref{eq:abstraction-exp-stability-kappabar}. Otherwise, the statement \eqref{eq:abstraction-exp-stability-kappabar} is trivially true and there is nothing to prove.
	Also assume $|x_0| \le r_0 \alpha^{-1} \tilde{\alpha}^{-1}$, and that the abstraction \eqref{eq:abstraction-guarantee}, \eqref{eq:abstraction-dynamics} for $x_{k+1}$ holds for $k=i$ if $|x_0| \le r_0$.

	\emph{Induction Assumption (IA):} Assume $|x_{0,\dots,i}| \leq r_0$.

	\emph{Induction Step $i$:}
	Under this assumption, the abstraction \eqref{eq:abstraction-guarantee}, \eqref{eq:abstraction-dynamics} holds for $k=0,\dots,i$. Therefore, the stability result \eqref{eq:abstraction-exp-stability-kappabar} derived from the abstraction also holds for $k=0,\dots,i+1$. (Note that this step is rather informal; for a rigorous proof, \eqref{eq:abstraction-exp-stability-kappabar} should be extended to abstractions which do not hold for all $k$ or all $x$.) This allows showing that (IA) also holds for $i+1$:
	\begin{align}
	\left.\begin{aligned}
	|x_0| &\leq r_0 (\alpha \tilde{\alpha})^{-1}\\
	|x_{0,\dots,i}| &\leq r_0
	\end{aligned}\right\rbrace  \quad \Rightarrow \quad	|x_{i+1}| \stackrel{\text{\eqref{eq:abstraction-guarantee}, \eqref{eq:abstraction-dynamics}, \eqref{eq:abstraction-exp-stability-kappabar}}}{\le} \alpha \tilde \alpha \underbrace{\tilde{\rho}^{i+1}}_{\mathrlap{\hspace{-.3em}\le 1 \text{ by assumption}}} r_0 (\alpha \tilde{\alpha})^{-1} \le r_0 \label{eq:abstraction-nonglobal-exp-proof-step}
	\end{align}

\emph{Conclusion:} (IA) is satisfied for $i=0$. By induction, $(IA)$ holds for all steps $i$, so by \eqref{eq:abstraction-nonglobal-exp-proof-step} the abstraction \eqref{eq:abstraction-guarantee}, \eqref{eq:abstraction-dynamics} and therefore the stability result $(64)$ hold for all times $k$.
\end{proof}

\section{Application to other types of weak execution} \label{sec:other-weak}
In practice, weak execution can go beyond the classical $(m,K)$-model. Convergence rate abstractions can be readily adapted to other scenarios, which we demonstrate in this section.

\paragraph{Extension to weak timing requirements} While the previous section only considers binary scheduling decisions $\sigma \in \lbrace 0,1 \rbrace$, the results can be extended to more complex scenarios, including non-binary integer or real-valued decisions.

One particular example is weak timing for sensors and actuators: As discussed in \cite{Ulbrich2019}, it may be expensive or impractical to guarantee strictly periodic timing for sampling all sensors and updating all actuators. Allowing some timing deviation $\Delta t$ is desirable for an efficient implementation, however it must be ensured that the system performance (stability) is still acceptable. If fixed worst-case bounds for the timing are determined as in \cite{Gaukler2019extended}, these may be very small. More flexibility can be gained by incorporating the history of the timing in an abstraction: Occasional large timing deviations are acceptable if the timing inbetween is good.

In this case, $\sigma_k=\Delta t_k$ is the vector of timing deviations for each sensor and actuator, and $\rho_{\sigma}$ a bound which depends on the timing $\sigma$. Future research will be concerned with determining this bound via quadratic Lyapunov functions based on \cite{Gaukler2019extended}.

\paragraph{Extension to disturbance}
If an unknown but bounded disturbance is present, the dynamics \eqref{eq:abstraction-vbar-dynamics} of the $\bar v$-abstraction are modified to upper bound the influence of disturbance:
\begin{equation}
	\bar v_{k+1} = \rho_{\sigma_k} \bar v_k + \beta \bar w_k, \quad \bar w_k \ge |w_k|
\end{equation}
If the disturbance amplitude is greatly varying, an upper bound of the current disturbance may possibly be estimated using interval disturbance observers as proposed  by \cite{Chakrabarty2017}, which will be examined in future work. Otherwise, it should be enough to consider a fixed upper bound.

If the execution $\sigma_k$ is then chosen such that $v_k \le C=\mathrm{const}$, this guarantees practical stability, \ie there are sets $X_0, S$ such that
\begin{equation}
	\forall k,~ \forall x_0 \in X_0: \quad x_k \in S.
\end{equation}
Particularly, $v_k \le C$ is equivalent to $S=\mathcal{B}_C$ due to the abstraction guarantee \eqref{eq:abstraction-guarantee}, and $X_0$ must be chosen small enough such that there is always a choice which permits $v_k \le C$.

\section{Application to control-aware dynamic scheduling} \label{sec:application-dyn_sched}

The focus of the previous sections was on design-time stability guarantees for fixed design parameters, \eg, fixed $(m,K)$. However, changing disturbance and execution conditions typically require that the parameters are chosen pessimistically for worst-case disturbance and timing. While this guarantees worst-case safety, in the average case it will typically be unnecessarily strict and therefore inefficient, in particular if the average case is significantly better than the worst case.

The key to solving this design conflict between safety and efficiency is run-time adaptation~\cite{Ulbrich2019}: In a similar way as feedback control reacts to the environment, adaptive real-time scheduling adapts the timing requirements to changing disturbance and system load. For example, for worst-case disturbance, the controller must be executed strictly, whereas otherwise, it may be skipped from time to time.

The main difficulty with adaptive scheduling of real-time control is the overhead introduced by scheduling decisions. In this section, we present an outlook on how convergence rate abstractions can be used to construct low-overhead adaptive scheduling.

\paragraph{Exponential Stability Without Disturbance:} If the disturbance is zero, $ \kappa_{0,k}$ can be easily computed or overapproximated to obtain information about the current quality of control, here in the sense of a decay rate, and to predict which quality would result from a certain scheduling decision in the real-time operating system.

To guarantee a specified worst-case decay $\hat \rho$ and overshoot $\hat \alpha$, compute $\hat { \kappa}_k \coloneqq \hat{\rho}^{-k}  \kappa_{0,k}$ by
\begin{align}
\hat \kappa_0 = 1, \qquad \hat { \kappa} _{k+1} = \hat{\rho}^{-k-1} \kappa_{0,k+1} \stackrel{\text{\eqref{eq:abstraction-kappa-def}}}{=} \hat{\rho}^{-k-1} \rho_{\sigma_k} \kappa_{0,k} =  \hat{\rho}^{-1}  \underbrace{ \rho_{\sigma_k}}_{\mathclap{\text{depends on scheduling decision and timing}}}   \hat { \kappa} _{k}, \qquad k \geq 0
\end{align}
and ensure that all scheduling decisions for the $k$-th period respect $\hat{\kappa}_{k}\le\hat{\alpha}$:
\begin{align}
	\forall k \geq 0: \quad \hat { \kappa}_k &= \hat{\rho}^{-k}  \kappa_{0,k} \le \hat \alpha \\
	\quad  \Leftrightarrow  \quad 	\forall k \geq 0: \quad \kappa_{0,k} &\le \hat \alpha \hat{\rho}^{k} \quad \stackrel{\text{\eqref{eq:abstraction-exp-stability-kappa}}}{\Leftrightarrow} \quad \text{exp. stable with }\hat\rho, \hat\alpha.
\end{align}

This results in a generalization of the $(m,K)$-scheme.
For example, if the sensor/actuator timing deviation was low in the recent past, the computation of $v_k$ will show that $|x_k|$ is low and it is okay to allow a large one-time deviation or even skip the controller once. This will increase flexibility or save energy and computational resources. On the other hand, after bad timing or skipping multiple times, the scheduling will return to mostly nominal execution. Therefore, the new scheme can be seen as a generalization of the $(m,K)$-scheme to a variable length $K$ and non-binary decisions.

\paragraph{Practical Stability With Disturbance:} In the presence of disturbance, the abstraction $v_k$ or $\bar v_k$ (or an upper bound) can be computed at run-time. This yields a bound on $|x_k|$, which can be interpreted as quality of control if $x_k=0$ is defined as the setpoint (or the state is transformed appropriately). Computing a prediction of the future abstraction value can be used to obtain scheduling decisions which guarantee an upper bound for $|x_k|$, \eg, that a quadrotor UAV does not fly too far away from its intended position.

\paragraph{Safety Supervisor:}
Predictive computation of $v_k$ can also be used to implement the supervisor suggested in \cite{Ulbrich2019}, which raises an alarm if scheduling is about to violate the specified quality of control. If the alarm is triggered, the system can switch back to a deterministic safety mode which guarantees nominal execution.

\paragraph{Connection to Quadratic Control Cost:} For simplicity, this paper discussed abstractions for $|x_k|$. However, the results can be extended to other measures of the quality of control. One important example is the quadratic control cost $J_k := x_k^\transp Q x_k$ with the symmetric and positive semi-definite weight matrix $Q$. This quantity can be bounded by
\begin{equation}
	J_k \leq C v_k^2, \; C = \max_{i} \lambda_i\{Q\}\geq 0,
\end{equation}
so the benefits of abstractions are similarly applicable to this case. If the weighting of state components in $J_k$ is inequal, this bound may be pessimistic, which can be reduced by transforming the state to \enquote{cost-like coordinates} before the abstraction is computed. For example, if $Q$ is positive definite, the Cholesky decomposition $Q=R^\transp R$ and the transformation $\tilde x := R x$ leads to $J_k = \tilde x_k^2$, so an abstraction $\tilde v$ of the $\tilde x$-dynamics directly tracks $J_k \le \tilde v_k^2$.

\begin{proof}[Proof for $J_k \le C v_k^2$]
Since $Q$ is real and symmetric it can be decomposed into $\Lambda = V^\transp Q V$ using its orthonormal eigenbasis $V^\transp V = V V^\transp = I$ and $\Lambda = \diag \{ \lambda_i\{Q\} \}$ (cf. \cite[p. 325, 2. and p. 283, 9.]{Bronstein16}). Herein, $\lambda_i\{Q\}$ are the eigenvalues of $Q$ which are all real and nonnegative. Bounding $J_k$ is then done by applying the bijective transform $z := V^\transp x \Leftrightarrow x = V z$:

\begin{align}
  J_k &= x_k^\transp Q x_k = z_k^\transp \underbrace{V^\transp Q V}_{\Lambda} z_k = \sum\limits_i \lambda_i\{Q\} z_{k,i}^2 \nonumber \\
  &\leq \underbrace{\max\limits_{i} \lambda_i\{Q\}}_{:= C} \; \sum\limits_i z_{k,i}^2 = C \; x_k^\transp \underbrace{V V^\transp}_{I} x_k = C |x_k|^2 \overset{\eqref{eq:abstraction-guarantee}}{\leq} C v_k^2
\end{align}
\end{proof}

\pagebreak
\section{Conclusion} \label{sec:conclusion}

The stability analysis of weakly-hard real-time control systems is significantly more complex than for the classical hard real-time case. In this paper, we propose \emph{convergence rate abstractions} as a method for reducing said complexity. At first, we formalized the approach in \cref{sec:formalization}: We characterize a dynamic system in state-space representation by means of its state radius to derive a one-dimensional model of the worst-case behavior. This time-varying bound proves to be a useful compromise between the pessimism inherent to static worst-case stability analysis and the complexity associated with analyzing the original model for weakly-hard execution.

We then showed that this abstraction is capable of incorporating different weak execution paradigms (\cref{sec:general-weak,sec:mk-weak,sec:other-weak}). The so obtained model does not only allow for deriving sufficient stability criteria but also for predicting a worst-case state bound given information on the external disturbance and the controller's timing properties. We then proposed some thoughts on using this information for control-aware adaptive scheduling in \cref{sec:application-dyn_sched}.

Summarizing, convergence rate abstractions contribute to stability analysis of weakly-hard real-time control systems in that they provide an intermediate layer which extracts static convergence properties from the ideal closed loop and combines them with the dynamic aspect of disturbance, both physical and originating from timing uncertainty.

In future work we aim to use this approach for provably safe adaptive real-time scheduling of control systems with little run-time overhead. This is enabled by the safety guarantees of an abstraction, as well as the simplicity of the one-dimensional abstraction dynamics. Additionally, we aim to use abstractions to complement existing heuristic techniques with a safety guarantee, and thereby overcome the classical design conflict between good average-case performance and provable worst-case stability.

We hope that the presented concept simplifies the analysis of real-time weakly-hard control systems and stimulates further research. Comments and feedback are highly appreciated.

\newpage

\printbibliography

\end{document}